%% file: main.tex
\documentclass[11pt]{article}
\usepackage{fullpage}
\usepackage{rpmacros}
\usepackage[T1]{fontenc}
\usepackage{gitinfo}
\usepackage[usenames,dvipsnames]{color}

\RequirePackage[colorlinks=true]{hyperref}
\hypersetup{
  linkcolor=[rgb]{0.3,0.3,0.6},
  citecolor=[rgb]{0.2, 0.6, 0.2},
  urlcolor=[rgb]{0.6, 0.2, 0.2}
}
\usepackage{fourier}
\usepackage{bm}
\usepackage{todonotes}
\usepackage{lipsum}
\usepackage{setspace}
\usepackage{scrextend}

\input{thmmacros}

\input{lazymacros} 
\input{bibmacros}
\usetikzlibrary{decorations.pathreplacing,arrows}

\onehalfspace

\title{Quasi-polynomial Hitting Sets for Circuits\\with Restricted Parse Trees}
\author{
  Ramprasad Saptharishi\thanks{Research supported by Ramanujan Fellowship of DST.} \quad \quad \quad Anamay Tengse\thanks{Supported by a fellowship of the DAE.}\\\small{Tata Institute of Fundamental Research, Mumbai, India}\\{\small \tt{\{ramprasad , tengse.anamay\}@tifr.res.in}}%
}
\begin{document}
\maketitle

{\let\thefootnote\relax
\footnotetext{\textcolor{white}{Cool! You found it! (FWIW) Base version:~(\gitAuthorIsoDate)\;,\;\gitAbbrevHash\;\; \gitVtag}}
}
\begin{abstract}
We study the class of non-commutative \emph{Unambiguous circuits or Unique-Parse-Tree (UPT) circuits}, and a related model of \emph{Few-Parse-Trees (FewPT)} circuits (which were recently introduced by Lagarde, Malod and Perifel~\cite{LMP16} and Lagarde, Limaye and Srinivasan~\cite{LLS17}) and give the following constructions:
  \begin{itemize}\itemsep0pt
  \item An explicit hitting set of \emph{quasipolynomial} size for UPT circuits,
  \item An explicit hitting set of \emph{quasipolynomial} size for FewPT circuits (circuits with constantly many parse tree shapes),
  \item An explicit hitting set of \emph{polynomial} size for UPT circuits (of known parse tree shape), when a parameter of \emph{preimage-width} is bounded by a constant.
  \end{itemize}
  The above three results are extensions of the results of \cite{AGKS15}, \cite{GKST15} and \cite{GKS16} to the setting of UPT circuits, and hence also generalize their results in the commutative world from \emph{read-once oblivious algebraic branching programs (ROABPs)} to \emph{UPT-set-multilinear} circuits.

  The main idea is to study \emph{shufflings} of non-commutative polynomials,  which can then be used to prove suitable depth reduction results for UPT circuits and thereby allow a careful translation of the ideas in \cite{AGKS15}, \cite{GKST15} and \cite{GKS16}. 
\end{abstract}

\section{Introduction}

The field of algebraic complexity deals with classifying multivariate polynomials based on their hardness.
Typically, the complexity of a polynomial is measured by the size  of the smallest circuit computing it
(an arithmetic circuit is a directed acyclic graph made up of internal nodes that are labeled with $+$ or $\times$ and leaves labelled with variables or constants from the field).
The central question in this field is to construct an explicit family of polynomials ($\set{\operatorname{Perm}_n}$ is the top candidate) that requires large arithmetic circuits to compute it.
This is also called the ``$\VP$ vs $\VNP$'' question (named after Valiant~\cite{V79}), and thought of as an algebraic analogue of the ``$\P$ vs $\NP$'' question. 

So far, the best lower bound we have for general arithmetic circuits computing an $n$-variate degree $d$ polynomial is a barely super-linear $\Omega(n\log d)$ lower bound by Baur and Strassen~\cite{BS83}.
Recent research has focused on proving lower bounds for restricted classes of circuits, either by bounding the depth of such circuits or by focusing on other syntactic restrictions.
One such syntactic restriction is to consider \emph{non-commutative circuits}, where we assume that the underlying variables $x_1,\ldots, x_n$ do not commute.
In the non-commutative model, there is an inherent order in which elements are multiplied and this adds restrictions on the way monomials can be computed ($xy \neq yx$ here and hence $x^2 + 2x y + y^2 \neq (x+y)^2 = x^2 + xy + yx + y^2$).
It is therefore natural to expect that it should be easier to prove lower bounds in this model.

Nisan~\cite{N91} introduced the non-commutative model, specifically the non-commutative algebraic branching programs (ABP). In his seminal paper, he showed that the non-commutative versions of the determinant and permanent polynomials (among others) require exponential sized non-commutative ABPs to compute them. In fact, using his technique, one could even reconstruct the smallest non-commutative ABP given just oracle access to that polynomial (cf. \cite{KS06})! Although we have exponential lower bounds for non-commutative ABPs, we do not have any non-trivial lower bounds for non-commutative circuits. Hrube\v{s}, Wigderson and Yehudayoff~\cite{HWY10} presented an approach via \emph{sum-of-squares} lower bounds but we do not have any non-trivial lower bounds for the class of general non-commutative circuits.

Limaye, Malod and Srinivasan~\cite{LMS16} extended Nisan's lower bound to non-commutative skew circuits, which are circuits where every multiplication gate has at most one child that is a non-leaf. Lagarde, Malod and Perifel \cite{LMP16} initiated the study of non-commutative \emph{unambiguous circuits}, or \emph{Unique Parse Tree (UPT)} circuits. These circuits, and generalizations are the main models of study in this paper. 

Arvind and Raja~\cite{AR16} also studied lower bounds for various subclasses of commutative set-multilinear circuits.  Some of the models they study also include analogues of UPT and FewPT circuits. They also proved lower bounds for UPT and FewPT set-multilinear circuits, and also for other subclasses of set-multilinear circuits called \emph{narrow} set-multilinear circuits, \emph{interval} set-multilinear circuits, the latter of which assumes the sum-of-squares conjecture of Hrube\v{s}, Wigderson and Yehudayoff~\cite{HWY10}.

\subsection{The model of study}

A parse tree of a circuit is obtained by starting at the root, and at every $+$ gate choosing exactly one child, and at every $\times$ gate choosing all its children (formally defined in \autoref{defn:parse-trees}). Informally, a parse tree of a circuit is  basically a \emph{certificate} of computation of a monomial in a circuit. Lagarde, Malod and Perifel \cite{LMP16} introduced a subclass of non-commutative circuits called \emph{Unique Parse Tree (UPT) circuits} or \emph{unambiguous circuits} where all parse trees of the circuit have the same shape (formally defined in \autoref{defn:UPT-FewPT}). The class of non-commutative UPT circuits subsumes the class of non-commutative ABPs as any ABP can be expressed as a left-skew circuit. A related model of \emph{set-depth-$\Delta$ formulas} was studied by Agrawal, Saha and Saxena~\cite{ASS13} that is a subclass of UPT circuits where the underlying parse trees are extremely regular\footnote{the formula is levelled, and all nodes at a level have the same fan-in}. 

Lagarde, Malod and Perifel \cite{LMP16} extended the techniques of Nisan \cite{N91} to give exponential lower bounds for UPT circuits. Subsequently, Lagarde, Limaye and Srinivasan~\cite{LLS17} extended the lower bounds to the class of circuits with parse trees of not-too-many shapes (at most $2^{o(n)}$ shapes).

In \autoref{fig:UPT-example}, (a) is an example of a UPT circuit with (b) being the underlying parse tree shape; (c) is an example of a circuit with two distinct parse tree shapes. 

\begin{figure}[h]
\begin{center}
  \begin{tikzpicture}[transform shape, scale=0.7]
    \node[draw, circle] (root) at (0,0) {$+$};
    \node[draw, circle] (m1) at (-2,-1.5) {$\times$}
    edge[->] (root);
    \node[draw, circle] (m2) at (0,-1.5) {$\times$}
    edge[->] (root);
    \node[draw, circle] (m3) at (2,-1.5) {$\times$}
    edge[->] (root);
    
    \node[draw, circle] (s1) at (-2.5,-3) {$+$}
    edge[->] (m1)
    edge[->] (m2);
    
    \node[draw, circle] (s2) at (0,-3) {$+$}
    edge[->] (m3);
    
    \node[draw] (x1) at (-3,-4.5) {$x_1$}
    edge[->] (s1)
    edge[->] (s2);
    \node[draw] (x2) at (-1,-4.5) {$x_2$}
    edge[->] (s1)
    edge[->] (s2)
    edge[->] (m1);
    \node[draw] (x3) at (1,-4.5) {$x_3$}
    edge[->] (s2);
    \node[draw] (x4) at (3,-4.5) {$x_4$}
    edge[->] (m2)
    edge[->] (m3);

    \node at (0,-5) {(a)};
    
    \begin{scope}[shift={(5,0)}]
      \node[draw, circle] (root) at (0,0) {$+$};
      \node[draw, circle] (m1) at (0,-1.5) {$\times$}
      edge[->] (root);
      \node[draw, circle] (s1) at (-0.75,-3) {$+$}
      edge[->] (m1);
      \node[draw, rectangle, fill=black] (l1) at (-1,-4.5) {}
      edge[->] (s1);
      \node[draw, rectangle, fill=black] (l1) at (1,-4.5) {}
      edge[->] (m1);

      \node at (0,-5) {(b)};
    \end{scope}

    \begin{scope}[shift={(13,0)}]
      \node[draw, circle] (root) at (0,0) {$+$};
      \node[draw, circle] (m1) at (-2,-1.5) {$\times$}
      edge[->] (root);
      \node[draw, circle] (m2) at (0,-1.5) {$\times$}
      edge[->] (root);
      \node[draw, circle] (m3) at (2,-1.5) {$\times$}
      edge[->] (root);
      \node[draw, circle] (s1) at (-2.5,-3) {$+$}
      edge[->] (m1)
      edge[->] (m2);
      \node[draw, circle] (s2) at (0,-3) {$+$}
      edge[->] (m2)
      edge[->] (m3);
      
      \node[draw] (x1) at (-3,-4.5) {$x_1$}
      edge[->] (s1)
      edge[->] (s2);
      \node[draw] (x2) at (-1,-4.5) {$x_2$}
      edge[->] (s1)
      edge[->] (m1)
      edge[->] (s2);
      \node[draw] (x3) at (1,-4.5) {$x_3$}
      edge[->] (s2);
      \node[draw] (x4) at (3,-4.5) {$x_4$}
      edge[->] (m3);

      \node at (0,-5) {(c)};
    \end{scope}
  \end{tikzpicture}
\end{center}
\caption{Examples of circuits with restricted parse trees}
\label{fig:UPT-example}
\end{figure}
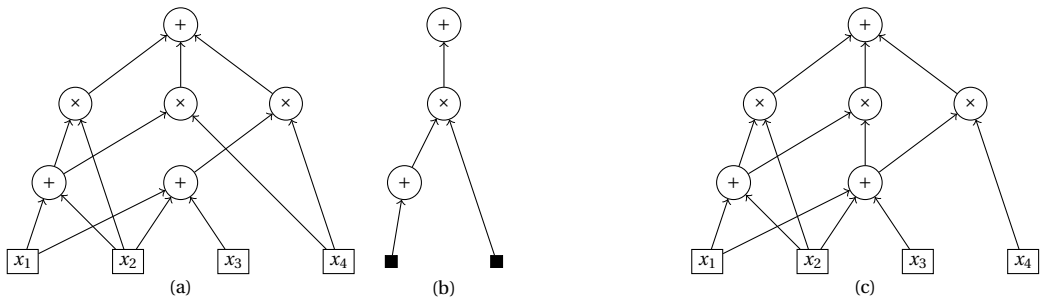

\subsection{Polynomial identity testing}

A \emph{Polynomial Identity Test (PIT)} is an algorithm that, given a circuit as input, checks if the circuit is computing the zero polynomial or not.
The standard Ore-DeMillo-Lipton-Schwartz-Zippel lemma \cite{O22,DL78,S80,Z79} provides a simple randomized algorithm but the goal is to construct an efficient deterministic PIT.
A stronger test is what is called a \emph{black-box PIT} where we are only provided evaluation access to the circuit.
Hence, a black-box PIT is essentially equivalent to constructing a \emph{hitting set} i.e., a set of points (or matrices, in the case of non-commutative polynomials) $\mathcal{H}$ such that every non-zero polynomial from the class of interest is guaranteed to evaluate to a nonzero value on some element $\veca \in \mathcal{H}$.
PITs that use the structure of the circuit are called \emph{white-box} PITs. 

The task of constructing efficient PITs is intimately connected to the task of proving lower bounds \cite{HS80, KI04,A05a}.
Once we have a lower bound for a class $\mathcal{C}$, it is natural to ask if we can also construct efficient PITs for that class.
Raz and Shpilka~\cite{RS05} gave the first deterministic polynomial time white-box PIT for the class of non-commutative ABPs.
Forbes and Shpilka~\cite{FS13} gave a quasipolynomial ($n^{O(\log n)}$) size hitting set for non-commutative ABPs.
This was achieved by studying a natural commutative analogue of non-commutative ABPs, and this was the class of \emph{Read-Once Oblivious Algebraic Branching Programs (ROABPs)} where the variables are read in a ``known order''.

The class of ROABPs is interesting in its own right owing to the connection with the ``$\mathsf{RL}$ vs $\mathsf{L}$'' question. In fact, much of the hitting set constructions for ROABPs has been inspired by  Nisan's~\cite{N92} pseudorandom generator for $\mathsf{RL}$ (which has seed length $O(\log^2 n)$). 
As mentioned earlier, Forbes and Shpilka gave a hitting set of size $n^{O(\log n)}$ for polynomial sized ROABPs when the order in which variables are read was known.
Agrawal, Gurjar, Korwar and Saxena \cite{AGKS15} presented a different hitting set for the class of commutative ROABPs that did not need the knowledge of the order in which the variables were read.
Subsequently, Gurjar, Korwar, Saxena and Thierauf~\cite{GKST15} studied polynomials that can be computed as a sum of constantly many ROABPs (of possibly different orders) and presented a polynomial time white-box PIT, and also a quasipolynomial time black-box PIT for this class.

Lagarde, Malod and Perifel~\cite{LMP16}, besides presenting lower bounds for non-commutative UPT circuits, also gave a polynomial time white-box PIT for this class. This was extended by Lagarde, Limaye and Srinivasan~\cite{LLS17} to a white-box algorithm for non-commutative circuits with constantly many parse tree shapes (analogous to the result of \cite{GKST15}). The question of constructing black-box PITs was left open by them, and we answer this in our paper. 

\subsection{Our results}

\subsubsection*{Polynomial Identity Testing}

Our main results are hitting sets for the class of polynomials computed by UPT circuits and related classes. 

\begin{theorem}[Hitting sets for UPT circuits]\label{thm:hitting-set-unamb}
There is an explicit hitting set $\mathcal{H}_{d,n,s}$ of at most $(snd)^{O(\log d)}$ size for the class of degree $d$ $n$-variate homogeneous non-commutative polynomials in $\F\inangle{x_1,\ldots, x_n}$ that are computed by UPT circuits of size at most $s$.
\end{theorem}

This result builds on the technique of \emph{basis isolating weight assignments} introduced by \cite{AGKS15} for constructing hitting sets for ROABPs. Furthermore, we can also extend the hitting set to the class of non-commutative circuits that have \emph{few shapes} (analogous to \cite{GKST15}'s hitting set for sum of few ROABPs).

\begin{restatable}[Hitting sets for circuits with few parse tree shapes]{theorem}{FewPTPIT}
  \label{thm:hitting-set-few-parse-trees}
There is an explicit hitting set $\mathcal{H}_{d,n,s,k}$ of size at most $(s^{2^k}nd)^{O(\log d)}$ for the class of $n$-variate degree $d$ homogeneous non-commutative polynomials in $\F\inangle{x_1,\ldots, x_n}$ that are computed by non-commutative circuits of size at most $s$ consisting of parse trees of at most $k$ shapes.
\end{restatable}

Both the above theorems are fully black-box in the sense that it is not required to know the underlying shape(s). For the case of non-commutative ABPs (and more generally, ROABPs in a known order), Gurjar, Korwar and Saxena \cite{GKS16} presented a more efficient hitting set when the width of the ABP is small. For UPT circuits, there is a natural notion of \emph{preimage-width} of a UPT circuit (formally defined in \autoref{defn:preimage-width}) that corresponds to the notion of width of an ABP. We show an analogue of the hitting set of Gurjar, Korwar and Saxena for the class of UPT circuits of small \emph{preimage-width} if the underlying shape of the parse trees is known. 

\begin{restatable}[Hitting sets for known-shape low-width UPT circuits]{theorem}{ConstantWidthUPT}
  \label{thm:hitting-set-low-width-known-shape}
  Let $\mathcal{C}_{n,d,T,w}$ be the class of $n$-variate degree $d$ non-commutative polynomials that are computable by UPT circuits of preimage-width at most $w$ and underlying parse-tree shape as $T$. Over any field of zero or large characteristic, there is an explicit hitting set $\mathcal{H}_{n,d,T,w}$ of size $w^{O(\log d)} \poly(nd)$ for $\mathcal{C}_{n,d,T,w}$. 
\end{restatable}

These hitting sets also translate to the natural commutative analogues of \emph{UPT set-multilinear circuits} etc. (formally defined in \autoref{defn:set-multilinear}). 

\subsubsection*{Structural results}

If $f$ is a non-commutative polynomial of degree $d$ and if $\sigma \in S_d$ is a permutation on $d$ letters, we define the \emph{shuffling} of $f$ by $\sigma$ (denoted by $\Delta_\sigma(f)$) as the natural operation of permuting each \emph{word} of $f$ according to $\sigma$. 

The three PIT statements stated above begin with the following depth reduction statement about UPT circuits.

\begin{restatable}[Depth reduction for UPT circuits]{theorem}{depthrednthm} \label{thm:depth-reduction}
  Let $f$ be an $n$-variate degree $d$ polynomial that is computable by a UPT circuit of preimage-width $w$. Then, there is some $\sigma \in S_d$ such that $\Delta_\sigma(f)$ can be computed by a UPT circuit of $O(\log d)$ depth and preimage-width at most $O(w^2)$. 
\end{restatable}

The above theorem implies that $\Delta_\sigma(f)$ is computable by an ABP of quasipolynomial size. We also show that this blow-up of quasipolynomial size is tight.

\begin{restatable}[Separating UPT circuits and ABPs, under shuffling]{theorem}{UPTABPsep}
\label{thm:HY-analogue} There is an explicit $n$-variate degree $d$ non-commutative polynomial $f$ that is computable by UPT circuits of preimage-width $w = \poly(n,d)$ such that for every $\sigma \in S_d$, the polynomial $\Delta_\sigma(f)$ requires non-commutative ABPs of size $(nd)^{\Omega(\log nd)}$ to compute it. 
\end{restatable}

We also extend the lower bound of \cite{LMP16} to give a polynomial computed by a \emph{skew circuit} that requires exponential sized UPT circuits under any shuffling. Details are in \autoref{appsect:exp-under-shuffling}.

\subsection{Proof ideas}

As mentioned, the starting point of all these results is the depth reduction. From a result of Nisan~\cite{N91}, the palindrome polynomial $\Pal_d$ is known to require ABPs of size $2^{\Omega(d)}$ even though it can be computed by a polynomial sized UPT circuit. Therefore, $\Pal_d$ cannot be computed by a circuit of depth $o(d/\log d)$. The key insight here is that even though $\Pal_d$ cannot be computed by small depth non-commutative circuits, a shuffling of the palindrome is
\[
  \sum_{w_1,\ldots, w_d \in [n]} x_{w_1}x_{w_1}x_{w_2}x_{w_2} \cdots x_{w_d}x_{w_d} = \prod_{i=1}^d \inparen{x_1x_1 + \cdots + x_nx_n},
\]
which is of course computable by an $O(\log d)$ depth UPT formula even. Hence we attempt to reduce the depth under a suitable shuffling. 

In order to establish the depth reduction (\autoref{thm:depth-reduction}) we follow the strategy of Valiant, Skyum, Berkowitz and Rackoff~\cite{VSBR83} and Allender, Jiao, Mahajan and Vinay~\cite{AJMV98} but make use of the UPT structure (work with different \emph{frontier nodes} and \emph{gate quotients}) based on the underlying shape of the parse trees. It was pointed out to us that the key ideas in our proof of depth reduction were used by Arvind and Raja~(\cite{AR16}) for a commutative analogue of UPT circuits.

This depth reduction immediately yields that there is a quasipolynomial sized ABP computing a shuffling of $f$. We show that this blow-up is tight (\autoref{thm:HY-analogue}) by essentially following the proof of Hrube\v{s} and Yehudayoff~\cite{HY16} to separate monotone ABPs and monotone circuits in the commutative world.

In order to obtain hitting sets for UPT circuits, one could potentially just use the fact that there is a quasipolynomial sized ABP computing a shuffling of $f$ and just use the known hitting sets for non-commutative ABPs~\cite{FS13} to obtain a hitting set of $\poly(ndw)^{O(\log^2 d)}$. However, we directly work with the UPT circuit and lift the technique of \emph{basis isolating weight assignments} of Agrawal, Gurjar, Korwar and Saxena~\cite{AGKS15} to this more general setting to obtain \autoref{thm:hitting-set-unamb}. \autoref{thm:hitting-set-low-width-known-shape} is a straightforward generalization of the ideas of Gurjar, Korwar and Saxena~\cite{GKS16} once we observe that the depth reduction keeps the preimage-width small.

\autoref{thm:hitting-set-few-parse-trees} essentially follows the same ideas of Gurjar, Korwar, Saxena and Thierauf~\cite{GKST15}. The techniques of \cite{GKST15} are general enough that once a circuit class has a \emph{characterizing set of dependencies} and a \emph{basis isolating weight assignment}, there is a natural method to lift the techniques to work with the sum of few elements from this class. \cite{GKST15} use this for ROABPs and we use this for UPT circuits. 

To summarize, once we obtain the depth reduction, much of the results in this paper is a careful translation of prior work of \cite{HY16}, \cite{AGKS15}, \cite{GKST15}, \cite{GKS16} to the setting of UPT (or FewPT) circuits. Consequently, this also generalizes the hitting sets of \cite{AGKS15, GKST15, GKS16} from ROABPs to \emph{UPT (or FewPT) set-multilinear} circuits. Such a generalization was unknown prior to this work. 

\section{Preliminaries}

\subsection{Notation}

\begin{itemize}\itemsep0pt
\item We  use $\F\inangle{x_1,\ldots, x_n}$ to refer to the ring of polynomials in non-commuting variables $\set{x_1,\ldots, x_n}$. For a parameter $d$, we use $\F\inangle{x_1,\ldots, x_n}_{\deg = d}$ to refer to the set of polynomials in $\F\inangle{x_1,\ldots, x_n}$ that are homogeneous and of degree $d$. Similarly, $\F\inangle{x_1,\ldots, x_n}_{\deg \leq d}$ refers to the set of polynomials of degree at most $d$.

\item We use boldface letters $\vecx$ and $\vecy$ to denote sets of variables (the number of variables would be clear from context). We shall also use $[d]$ to refer to the set $\set{1,2,\ldots, d}$. 

\item The paper would sometime shift between the commutative and the non-commutative domains. We use $\vecx$ whenever we are talking about non-commutative variables, and $\vecy$, $\vecz$ for variables in the commutative domain. 
\end{itemize}

\subsection{Basic definitions}

\subsubsection*{UPT and FewPT circuits}

\begin{definition}[Parse trees]\label{defn:parse-trees}
  A parse tree $T$ of a circuit $C$ is a tree obtained as follows:
  \begin{itemize}\itemsep0pt
  \item the root of $C$ is the root of $T$,
  \item if $v \in T$ is a $\times$ gate, then all the children in $C$ are the children of $v$ in  $T$ in the same order,
  \item if $v\in T$ is a $+$ gate, then exactly one child of $v$ in $C$ is a child of $v$ in $T$. 
  \end{itemize}
  The value of the parse tree $T$, denoted by $[T]$,  is just the product of the leaf labels in $T$. 
\end{definition}

Intuitively, a parse tree is a \emph{certificate} that a monomial was produced in the computation of $C$ (though it could potentially be canceled by other parse trees computing the same monomial). Therefore, if $f$ is the polynomial computed by $C$, then
\[
  f = \sum_{T\text{ is a parse tree}} [T].
\]

\begin{definition}(UPT and FewPT circuits)\label{defn:UPT-FewPT}
  A circuit $C$ computing a homogeneous polynomial is said to be a \emph{Unique Parse Tree (UPT) circuit} if all parse trees of $C$ have the same shape (that is, they are identical except perhaps for the gate names).

  A circuit $C$ that computes a homogeneous polynomial is said to be a FewPT$(k)$ circuit if the parse trees of $C$ have at most $k$ distinct shapes. 
\end{definition}

\begin{definition}[Preimage-width]\label{defn:preimage-width}
  Suppose $C$ is a UPT circuit and say $T$ is the shape of the underlying parse trees. For a node $\tau \in T$ and a gate $g\in C$, we shall say that $g$ is a \emph{preimage} of $\tau$, denoted by $g \sim \tau$, if and only if there is some parse tree $T'$ of $C$ where the gate $g$ appears in position $\tau$.

  The \emph{preimage-width} of a UPT circuit $C$ is the largest size of preimages of any node $\tau \in T$. That is,
  \[
    \operatorname{preimage-width}(C) = \max_{\tau \in T} \abs{\setdef{g\in C}{g \sim \tau}}. \qedhere
  \]
\end{definition}

It is clear that if $C$ is a UPT circuit of preimage-width $w$ computing a homogeneous degree $d$ polynomial, then the size of $C$ is at most $dw$. The preimage-width of a UPT circuit is a more useful measure to study than the size of the circuit. A simple concrete example of this is that the standard conversion of homogeneous ABPs to homogeneous circuits in fact yields UPT circuits. Furthermore, the width of the ABP is directly related to the preimage-width of the resulting UPT circuit. 

\begin{observationwp}
  If $f$ is computable by a width $w$ homogeneous algebraic branching program, then $f$ can be equivalently computed by UPT circuits of preimage-width $w^2$. 
\end{observationwp}

\subsubsection*{$\times_p$-products}

\begin{definition}[$\times_p$-products]
  For any $d_1,d_2 \geq 0$ and $p$ satisfying $0\leq p \leq d_2$, define $\times_p$ as the unique bilinear map
  $\times_p: \F\inangle{x_1,\ldots, x_n}_{\deg = d_1} \times \F\inangle{x_1,\ldots, x_n}_{\deg = d_2} \rightarrow \F\inangle{x_1,\ldots, x_n}_{\deg = d_1+d_2}$  that satisfies
  \[
    x_{w_1}\cdots x_{w_{d_1}} \times_p x_{v_1}\cdots x_{v_{d_2}} = x_{v_1}\cdots x_{v_p}x_{w_1} \cdots x_{w_{d_1}} x_{v_{p+1}} \cdots x_{v_{d_2}}.\qedhere
  \]
\end{definition}
\noindent
For instance, the usual multiplication (or concatenation) operation is just $\times_{0}$.

\subsubsection*{Shuffling of a polynomial}

\begin{definition}[Shuffling of a non-commutative polynomial]
Let $P_d(x_1,\ldots, x_n) \in \F\inangle{x_1,\ldots, x_n}_{\deg = d}$ be a homogeneous degree $d$ non-commutative polynomial. Given any permutation $\sigma \in S_d$ over $d$-letters, we can define the \emph{shuffling of $P_d$ via $\sigma$} as the unique linear map $\Delta_\sigma:\F\inangle{x_1,\ldots, x_n}_{\deg = d}\rightarrow\F\inangle{x_1,\ldots, x_n}_{\deg = d}$ that is obtained by linearly extending
\[
  \Delta_\sigma(x_{w_1}\cdots x_{w_d}) = x_{w_{\sigma(1)}} \cdots x_{w_{\sigma(d)}}.\qedhere
\]
\end{definition}

\subsection{Basic lemmas}

\subsubsection*{Canonical UPT circuits, and types of gates}

We shall say that a UPT circuit $C$ with underlying parse tree shape $T$ is \emph{canonical} if for every gate $g\in C$ there is some node $\tau\in T$ such that every parse tree of $C$ involving $g$ has $g$ only in position $\tau$. In other words, every gate of the circuit has a unique type associated with it. 

\begin{lemma}[\cite{LMP16}]\label{lem:canonical-wlog} Suppose if $f \in \F\inangle{x_1,\ldots, x_n}$ is a homogeneous, degree $d$, non-commutative polynomial computed by a non-commutative UPT circuit of preimage-width $w$. Then, $f$ can be equivalently computed by a canonical UPT circuit of preimage-width $w$ as well. 
\end{lemma}

For a canonical UPT circuit where the parse trees have shape $T$, we shall say that $g$ has type $\tau$ if $\tau \in T$ is the unique node in $T$ such that $g \sim \tau$. 

Fix a $\tau \in T$ and let $i$ be the number of leaves of the subtree rooted at $\tau$, and let $p$ be the number of leaves to the left of $\tau$ in the inorder traversal of $T$.
We shall then say that $\tau$ (or a gate $g\in C$ of type $\tau$) has \emph{position-type} $(i,p)$. 
The following lemma allows us to write the polynomial computed by the circuit as a small sum of $\times_p$-products.

\begin{lemma}[\cite{LMP16}]\label{lem:parse-tree-to-p-product} Let $f$ be a polynomial computed by a canonical UPT circuit $C$ of preimage-width $w$ and say $T$ is the shape of the underlying parse trees. If $\tau \in T$ with position-type $(i,p)$, then we can write $f$ as
  \[
    f(\vecx) = \sum_{r=1}^w g_r(\vecx) \times_p h_r(\vecx),
  \]
where $\deg g_r = i$ and $\deg h_r = \deg(f) - i$ for all $r = 1,\ldots, w$. 
\end{lemma}

\section{Depth reduction for UPT circuits}
\label{sec:depth-reduction}

\noindent
This section shall address \autoref{thm:depth-reduction}, which we recall below. 

\depthrednthm*

It was pointed out to us that a very similar depth reduction was also proved by Arvind and Raja~\cite{AR16}. They showed that a commutative UPT set-multilinear circuit can be depth-reduced to a corresponding quasi-polynomial sized $O(\log d)$ depth UPT set-multilinar formula via Hyafil's \cite{H79} depth reduction. Using techniques similar to \cite{VSBR83}, one can obtain a polynomial sized circuit of depth $O(\log d)$ while maintaining unambiguity. Though this can be inferred from the results in \cite{AR16}, we state and prove it in the form needed for the non-commutative setting.

\subsection{UPT $\otimes$-circuits}

To prove the depth reduction, we will move to an intermediate model of UPT $\otimes$-circuits. 

\begin{definition}[UPT $\otimes$-circuits]\label{defn:otimes-circuits}
  The class of UPT $\otimes$-circuits is a generalization of homogeneous non-commutative circuits in that the internal gates are $+$ gates and $\times_p$ gates instead of the usual $+$ and $\times$ gates. We shall also say that the circuit is \emph{semi-unbounded} if all $\times_p$ gates have fan-in bounded by $2$ (with no restriction on $+$ gates). 

  A \emph{parse tree} for an $\otimes$-circuit is similar to parse trees in a general non-commutative circuit but the internal nodes of the parse tree are labelled by $+$ and $\times_p$ (with the $p$ specified at each gate).

  We shall say that an $\otimes$-circuit $C$ is UPT if every parse tree is of the same shape, i.e. two parse trees in $C$ can differ only in the gate names. 
\end{definition}

\noindent
To prove \autoref{thm:depth-reduction}, we shall first depth reduce the circuit to obtain an $\otimes$-circuit computing $f$ of $O(\log d)$ depth. Then, we will convert that to a UPT circuit that computes a shuffling of $f$. 

\begin{restatable}[Depth reducing to $\otimes$-circuits]{lemma}{DepthRedToOtimes}
  \label{lem:depth-red-to-otimes-ckts}
  Let $f \in \F\inangle{x_1,\ldots, x_n}$ be a homogeneous degree $d$ polynomial that is computable by a UPT circuit of preimage-width $s$. Then, $f$ can be equivalently be computed by a semi-unbounded UPT $\otimes$-circuit of preimage-width $O(s^2)$ and depth $O(\log d)$. 
\end{restatable}

\begin{proof}
Let $C$ be the UPT circuit computing $f(x_1,\ldots,x_n)$ and say $T$ is the shape of the parse trees of $C$. For any node $\tau \in T$, let $\mathcal{F}_\tau$ be the set of all gates in $C$ whose position in $T$ is $\tau$. 
For two gates $u,v \in C$, we shall say that $u \succeq v$ if the place of $u$ in $T$ is an ancestor of the place of $v$ in $T$. We shall abuse notation and use $u \succeq \tau$ to mean that $u$'s position in $T$ is an ancestor of $\tau \in T$. For a gate $u\in C$, let $[u]$ refer to the polynomial computed at that gate. Similar to \cite{VSBR83,AJMV98}, we define inductively the following notion of a \emph{gate quotient} for any pair of gates $u,v\in C$:
  \[
    [u:v] = \begin{cases}
      0 & \text{if $u \nsucceq v$},\\
      1 & \text{if $u = v$},\\
      [u_1: v] + [u_2 : v] & \text{if $u = u_1 + u_2$},\\
      [u_1: v] \cdot [u_2] & \text{if $u = u_1\times  u_2$ and $u_1 \succeq v$},\\
      [u_1] \cdot [u_2:v] & \text{if $u = u_1\times  u_2$ and $u_2 \succeq v$}.      
    \end{cases}
  \]
  \indentBlock{
    \begin{claim}\label{claim:vsbr-recursion-unamb}
      For any $u \in C$, if $\tau \in T$ such that $u \succeq \tau$, then
      \begin{equation}\label{eqn:vsbr-u}
        [u] = \sum_{\substack{w \in C\\ w \sim \tau}} [w]\times_p [u:w]
      \end{equation}
      for a suitable $p$ depending just on $\tau$ and the type of $u$. 
      Furthermore, suppose $u,v\in C$ with $v$ being a multiplication gate and if $\tau \in T$ such that $u \succeq \tau \succeq v$ then
      \begin{equation}\label{eqn:vsbr-uv}
        [u:v] = \sum_{\substack{w \in C\\ w \sim \tau}} [w:v] \times_p [u:w]. 
      \end{equation}
      for a suitable $p$ depending just on $\tau$ and the type of $u$ and $v$. 
    \end{claim}
  }

  \medskip
  
  \noindent
  We'll defer this proof to later and first finish the proof of \autoref{lem:depth-red-to-otimes-ckts}. 
  With \eqref{eqn:vsbr-u} and \eqref{eqn:vsbr-uv}, we can construct the $\otimes$-circuit $C'$ for $f$ just as in \cite{VSBR83,AJMV98}. The circuit $C'$ would have gates computing each $[u]$ and $[u:v]$ for nodes $u,v\in C$ with $u \succeq v$ and $v$ being a multiplication gate. The wirings in $C'$ is built by appropriate applications of \eqref{eqn:vsbr-u} and \eqref{eqn:vsbr-uv}.

  Let $u \in C$ and say $\deg [u] = d_u$. The plan would be to set up the computation in $C'$ so that using an $O(1)$ depth computation, we can compute $[u]$ using gates whose degrees are a constant factor smaller than $d_u$. Consider any parse tree rooted at $u$, and starting from $u$ follow the \emph{higher degree} child. Let $\tau$ be the last point on the path with degree $\geq d_u/2$ (degree of its children will be $< d_u/2$). Applying \eqref{eqn:vsbr-u},
  \begin{align*}
    [u] & = \sum_{w\sim \tau} [w] \times_p [u:w]\\
        & = \sum_{w \sim \tau} ([w_1] \times [w_2]) \times_p [u:w] &\text{where $w = w_1 \times w_2$.}
  \end{align*}
  Now observe that each of the terms on the RHS, $[u:w], [w_1],[w_2]$ have degree at most $d_u/2$, as we wanted. Furthermore, each coordinate of tuple $([u:w], [w_1], [w_2])$ are all of the same \emph{type} as we run over all $w \sim \tau$.

We now need to show how to compute $[u:v]$ for a pair $u \succ v$. Say $\deg [u] = d_u$ and $\deg [v] = d_v$. For this, start with some parse tree rooted at $u$ and walk down the path leading to the place of $v$, and let $\tau$ be the last point on this path such that $\deg \tau \geq \frac{d_u + d_v}{2}$. Using \eqref{eqn:vsbr-uv},
  \begin{align*}
    [u:v] & = \sum_{w\sim \tau} [w:v]\times_p [u:w]\\ & = \sum_{w \sim \tau} \inparen{[w_1]\times [w_2:v]} \times_p [u:w] \end{align*} where $w = w_1 \times w_2$ and $w_2 \succeq v$ (the other possibility is identical).
By the choice of $\tau$, we have $\deg [u:w], \deg[w_2:v] \leq \frac{d_u - d_v}{2}$.
However, the best bound we can give on $\deg [w_1]$ is $d_u - d_v$.
Nevertheless, we can apply \eqref{eqn:vsbr-u} again on $[w_1]$ by finding a suitable $\tau' \prec w_1$ satisfying $\deg \tau' \geq \frac{\deg w_1}{2}$ and write
      \begin{align*}
        [u:v] & = \sum_{w \sim \tau} \inparen{[w_1]\times [w_2:v]} \times_p [u:w] \\
              & = \sum_{w \sim \tau} \inparen{\inparen{\sum_{w' \sim \tau'} [w'] \times_{p'} [w_1:w'] }\times [w_2:v]} \times_p [u:w]\\
              & = \sum_{w \sim \tau}\sum_{w' \sim \tau'} \inparen{\inparen{\inparen{[w_1']\times [w_2']} \times_{p'} [w_1:w'] }\times [w_2:v]} \times_p [u:w]
      \end{align*}
      By the choice of $\tau$ and $\tau'$, each of the \emph{factors} on the RHS have degree at most $\frac{(d_u - d_v)}{2}$ as we wanted. Furthermore, once again, all of the summands consists of similarly typed factors.

      This naturally yields an $\otimes$-circuit computing $f$ of depth $O(\log d)$ and size $\poly(s)$. Since all summands consist of similarly typed factors, it follows that the circuit is UPT as well.
    \end{proof}

\begin{proofof}{\autoref{claim:vsbr-recursion-unamb}}
  The proof is by induction. As a base case, suppose $u \sim \tau$. Then, $[u]$ is just the sum of the values of parse trees. Some of the parse trees use $u$. Of all nodes $w \in C$ such that $w\sim \tau$, only $[u:u] = 1$ and every other $[u:w] = 0$. Therefore, clearly $[u] = \sum_{w\sim \tau} [w] \cdot [u:w] $.
  
  Now suppose $u \succ \tau$ and say we already know that $[u'] = \sum_{w\sim\tau} [w]\times_p [u':w]$ for every $u\succ u'\succeq \tau$. If $u = u_1 + u_2$, then
  \begin{align*}
    [u] & = [u_1] + [u_2]\\
        & = \inparen{\sum_{w\sim \tau} [w]\times_p [u_1:w]} + \inparen{\sum_{w\sim \tau} [w]\times_p [u_2:w]}\\
        & = \sum_{w\sim \tau} [w]\times_p \inparen{[u_1:w] + [u_2:w]}\\
        & = \sum_{w\sim \tau} [w]\times_p [u:w].
  \end{align*}
  Similarly, suppose $[u] = [u_1] \times [u_2]$. We have two cases depending on whether $u_1 \succeq \tau$ or $u_2 \succeq \tau$. 

  \medskip

  \begin{minipage}{0.41\textwidth}
    If $u_1 \succeq \tau$, then
    \begin{align*}
      [u] & = [u_1] \times [u_2]\\
          & = \inparen{\sum_{w\sim \tau} [w] \times_p [u_1:w]} \times [u_2]\\
          & = \sum_{w\sim \tau} [w]\times_p \inparen{[u_1:w] \times [u_2]}\\
          & = \sum_{w\sim \tau} [w] \times_p [u:w].
    \end{align*}
  \end{minipage}%
  \hfill
  \begin{minipage}{0.46\textwidth}
    If $u_2 \succeq \tau$, then
    \begin{align*}
      [u] & = [u_1] \times [u_2]\\
          & = [u_1] \times \inparen{\sum_{w\sim \tau} [w]\times_p [u_2:w]}\\
          & = \sum_{w\sim \tau} [w]\times_{p+\deg u_1}\inparen{[u_1]\times [u_2:w]} &\\
          & = \sum_{w\sim \tau} [w]\times_{p+d_1} [u:w].
    \end{align*}
  \end{minipage}%
  
  \medskip

  \noindent
  Essentially the same proof works for \eqref{eqn:vsbr-uv} as well.
\end{proofof}


\begin{restatable}[$\otimes$-circuits to circuits for a shuffling]{lemma}{OtimesToShuffling}
  \label{lem:otimes-ckts-to-shuffling}
  Let $f \in \F\inangle{x_1,\ldots, x_n}$ be a homogeneous degree $d$ polynomial that is computable by a UPT $\otimes$-circuit $C'$ of size $s$. Consider the circuit $C''$ obtained by replacing all $\otimes$ gates in $C'$ by $\times$ gates. Then, $C''$ computes $\Delta_\sigma(f)$ for some $\sigma \in S_d$. 
\end{restatable}

\begin{proof}
  We shall prove this by induction. We need a slightly stronger inductive hypothesis which is that the choice of permutation $\sigma$ depends only on the shape of the parse trees in $C'$.

  Say $u$ is the root of $C'$. Suppose $u$ is a $+$ gate and say $u = u_1 + u_2 + \cdots + u_r$. If $u' = u_1' + \cdots + u_r'$ is the resulting computation in $C''$ then by the inductive hypothesis, we know that there is a $\sigma \in S_d$ such that $[u_i'] = \Delta_{\sigma}([u_i])$. Therefore,
  \[
    [u'] = \sum_{i=1}^r \Delta_{\sigma}([u_i])  = \Delta_{\sigma}([u]). 
  \]
  Suppose $u = u_1 \times_p u_2$ with $\deg [u_1] = d_1$ and $\deg [u_2] = d_2$. Say $u_1 = \sum_{\alpha \in [n]^{d_1}} a_\alpha x_\alpha$ and $\sum_{\beta \in [n]^{d_2}} b_\beta x_\beta$. Then, $[u] = \sum_{\alpha, \beta} a_\alpha b_\beta \cdot x_\alpha \times_p x_{\beta}$. If $u'$, $u_1'$ and $u_2'$ is the resulting computation in $C''$, then
  \begin{align*}
    [u'] &= [u_1'] \times [u_2']\\
         & = \Delta_{\sigma_1}([u_1]) \times \Delta_{\sigma_2}([u_2]) & \text{for some $\sigma_1 \in S_{d_1},\sigma_2\in S_{d_2}$,}\\
         & = \sum_{\alpha, \beta} a_\alpha b_\beta \cdot (\Delta_{\sigma_1}(x_\alpha) \times \Delta_{\sigma_2}(x_\beta))\\
         & = \sum_{\alpha, \beta} a_\alpha b_\beta \cdot \Delta_{\sigma}(x_\alpha \times_p x_\beta) & \text{for some $\sigma \in S_d$},\\
         & = \Delta_\sigma([u]) \qedhere
  \end{align*}
\end{proof}

\noindent
Together, \autoref{lem:depth-red-to-otimes-ckts} and \autoref{lem:otimes-ckts-to-shuffling} yield \autoref{thm:depth-reduction}. \hfill\qed(\autoref{thm:depth-reduction})\\

The following corollary is immediate from the fact that any circuit of depth $D$ and size $s$ can be computed by a formula of size $s^{O(d)}$ and hence an ABP of size $s^{O(d)}$. 

\begin{corollary}\label{cor:depth-reduction-uc-shuffing}
  If $f \in \F\inangle{x_1,\ldots, x_n}$ is a homogeneous degree $d$ polynomial that is computable by a UPT circuit of size $s$, then there is some $\sigma \in S_d$ such that $\Delta_\sigma(f)$ is computable by a non-commutative algebraic branching program of size $s^{O(\log d)}$.

  Furthermore, the shuffling $\sigma$ that permits this can also be efficiently computed given the underlying shape for the circuit computing $f$. 
\end{corollary}

\subsection{UPT circuits of  constant width}

For a UPT circuit $C$, we shall say that its \emph{width} is $w$ if for every node $\tau$ in the shape $T$, there are at most $w$ gates of $C$ that have type $\tau$. The following observation is evident from the proof of the above depth reduction.

\begin{observation}\label{obs:depth-reduction-width}
  If $C$ is a UPT circuit of width $w$, then the depth reduced circuit $C'$ as obtained in \autoref{thm:depth-reduction} has width $O(w^2)$.
\end{observation}

This observation would allow us to yield a more efficient hitting set for the class of \emph{small width known shape} UPT circuits. Details are present in \autoref{subsec:constantWidthUPT-appendix}. 

\section{Separating ROABPs and UPT circuits}

\UPTABPsep* 

The polynomial and the proof technique described here were introduced by Hrube\v{s} and Yehudayoff \cite{HY16} to separate monotone circuits and monotone ABPs in the commutative regime. The polynomial described here is a non-commutative analogue of the polynomial used by \cite{HY16}. Much of the proof is also the argument of \cite{HY16} tailored to the non-commutative setting. 

\subsection{The polynomial}

Let $T_d$ denote the complete binary tree of depth $d$ (with $2^d$ leaves) and let $D = 2^{d+1} - 1$ refer to the number of nodes in $T_d$. We shall say that a colouring $\gamma: T_d \rightarrow \Z_m$ is \emph{legal} if for every node $u\in T$, if $v$ and $w$ are the children of $u$ then $\gamma(u) = \gamma(v) + \gamma(w)\bmod{m}$.

Let $v_1,\ldots, v_D$ be the vertices of $T_d$ listed in an \emph{in-order} manner (left-subtree listed inductively, then the root, and then the right-subtree listed inductively). We now define the non-commutative polynomial $P_d(x_1,\ldots, x_m) \in \F\inangle{x_1,\ldots, x_m}$ of degree $D=2^{d+1}-1$ as
\begin{equation}\label{eqn:P_d-defn}
  P_d(x_1,\ldots, x_m) = \sum_{\substack{\gamma \in [m]^D\\\text{$\gamma$ is legal}}} x_{\gamma(v_1)} x_{\gamma(v_2)} \cdots x_{\gamma(v_D)}. 
\end{equation}

\begin{restatable}[Upper bound]{lemma}{UPTABPSepUB}
  \label{lem:UPT-ABP-sep-ub}
For every $m,d > 0$, the polynomial $P_d(y_1,\ldots, y_m)$ can be computed by a non-commutative UPT circuit of size $O(m^2 d)$.
\end{restatable}

\noindent
(Refer to \autoref{appsect:abp-upt-sep} for a proof).

\begin{restatable}[Lower bound]{theorem}{UPTABPSepLB}
  \label{thm:Pd-abp-lowerbound}
For every permutation $\sigma \in S_D$, any non-commutative ABP computing the polynomial $\Delta_\sigma(P_d)$ has width $m^{\Omega(d)}$. 
\end{restatable}

Hence for $d = \log m$, we have that $P_d(x_1,\ldots, x_m)$ is computable by a UPT circuit of size $O(m^2 \log m)$ but for every $\sigma \in S_D$ the above theorem tells us that $\Delta_\sigma(P_d)$ requires ABPs of width $m^{\Omega(\log m)}$ to compute it. The lower bound follows on exactly same lines as the \cite{HY16}. A proof is present in \autoref{appsect:abp-upt-sep}.

\section{Hitting sets for non-commutative models}
\label{sec:BIWAforUPT}

\subsubsection*{Commutative brethren of non-commutative models}

This reduction to an appropriate commutative case was used by Forbes and Shpilka~\cite{FS13} to reduce constructing hitting sets for non-commutative ABPs to hitting sets for commutative ROABPs (more precisely, to set-multilinear ABPs). They studied the image of the non-commutative polynomial under the map
$\Psi:\F\inangle{x_1,\ldots, x_n}_{\deg=d} \rightarrow \F[y_{1,1},\ldots, y_{d,n}]$ which is the unique $\F$-linear map given by $\Psi: x_{w_1}\cdots x_{w_d} \mapsto y_{1,w_1}\cdots y_{d,w_d}$.

For the model of non-commutative UPT circuits, the appropriate commutative model is a restriction of set-multilinear circuits that we call UPT set-multilinear ($\ComUPT$) circuits.

  \begin{definition}[Set-multilinear circuits]\label{defn:set-multilinear}
    Let $\vecy = \vecy_1 \sqcup \cdots \sqcup \vecy_d$ be a partition of the variables.
    A circuit $C$ computing a polynomial $f \in \F[\vecy]$ is said to be a \emph{set-multilinear circuit} with respect to the above partition if:
    \begin{itemize}\itemsep0pt
    \item each gate $g\in C$ is labelled by a subset $S_g \subseteq [d]$ and $g$ computes a polynomial over variables $\Union_{i\in S_g} \vecy_i$ where every monomial of $[g]$ is divisible by exactly one variable in $\vecy_i$ for each $i\in S_g$,
    \item if $g$ is a $+$ gate, then the subset that labels $g$ also labels each of its children,
    \item if $g$ is a $\times$ gate with $g_1$ and $g_2$ being its children, then the subsets $S_{g_1}$ and $S_{g_2}$ labelling $g_1$ and $g_2$ respectively is a partition of $S_{g}$, i.e. $S_{g} = S_{g_1} \sqcup S_{g_2}$. \qedhere
    \end{itemize}

	We shall say the circuit $C$ is \emph{UPT set-multilinear} if every parse tree of $C$ is of the same shape and identically labelled. That is, if $g$ and $g'$ are $\times$ gates labelled by a set $S \subseteq [d]$, and if $g = g_1 \times g_2$ with $S_1$ and $S_2$ labelling $g_1$ and $g_2$, then the children of $g'$ are also labelled by $S_1$ and $S_2$ respectively.

    We shall say the set-multilinear circuit $C$ is \emph{FewPT$(k)$ set-multilinear} if the circuit consists of parse trees of at most $k$ different shapes. 
  \end{definition}
A natural generalization that will be useful later is a \emph{multi-output UPT set-multilinear} circuit, which is a UPT set-multilinear circuit that potentially has multiple output gates, which are all labelled with the same subset.

Forbes and Shpilka~\cite{FS13} showed that constructing hitting sets for these commutative models suffices for the non-commutative models by a simple reduction (details in \autoref{appsubsect:defn-UPT-SML}). We shall therefore focus on these commutative models for the hitting set constructions. And since we have already seen that such circuits can be depth reduced\footnote{the shuffling just reorders the partition of the set-multilinear circuit} to $O(\log d)$ depth, it suffices to construct a hitting set for $O(\log d)$-depth UPT and FewPT set-multilinear circuits. 

\subsection{Preliminaries for PIT}

\subsubsection*{Weight assignments and basis isolation}

   To construct hitting sets for ROABPs, Agrawal, Gurjar, Korwar and Saxena~\cite{AGKS15} defined the notion of \emph{basis isolating weight assignments} for associated vector spaces of polynomials. The description presented here is an adaptation of the approach of \cite{AGKS15} to set-multilinear circuits of small depth.

\begin{definition}[Basis Isolating Weight Assignment (BIWA)]\label{defn:biwa} 
    A \emph{weight assignment} is a function  $\wt: \vecy \rightarrow [M]^k$, for some positive integer $M$, that can then be extended to all multilinear monomials over $\vecy$ via
  \[
    \wt\inparen{\prod_{i\in S}y_i} = \sum_{i\in S}^n \wt(y_{i}).\qedhere
  \]

  Let $V$ be a vector space of polynomials in $\F[\vecy]$, which can also be thought of as a matrix with a generating set of polynomials listed out as rows (with each column being indexed by a monomial in $\vecy$). 
    
    Such a weight assignment $\wt$ is said to be a \emph{basis isolating weight assignment} for $V$ if there exists a basis of its column space, indexed by $B \subseteq \Mons(\vecy)$, such that
    \begin{enumerate}
    \item if $m_1,m_2 \in B$ and $m_1\neq m_2$, then $\wt(m_1) \neq \wt(m_2)$,
    \item for every $m \notin B$,
      \[
        V_m \in \Fspan\setdef{V_{m'}}{m'\in B\;,\; \wt(m') \prec \wt(m)}
      \]
      where by $V_m$ we mean the column of $V$ indexed by the monomial $m$ and $\prec$ is the lexicographic ordering on $M^k \subset \N^k$. \qedhere
    \end{enumerate}
  \end{definition}
  
\begin{lemma}[\cite{AGKS15}]\label{lem:biwa-preserve-nonzero} Let $V$ be a vector space of polynomials in $\F[\vecy]$ and say $f\in V$. If $\wt:\vecy \rightarrow [M]^k$ is a BIWA for $V$, then if $\vect = \set{t_1,\ldots, t_k}$
    \begin{align*}
      f(y_1,\ldots, y_n) \neq 0 \Longleftrightarrow  & f(\vect^{\wt(y_1)}, \cdots, \vect^{\wt(y_n)}) \neq 0\\
     & \text{(where } \vect^{(\alpha_1,\ldots, \alpha_k)} \text{ is short-hand for $t_1^{\alpha_1} \cdots t_k^{\alpha_k}$ )}.
    \end{align*}
  \end{lemma}

If $f\neq 0$ and $\deg(f) \leq d$, then $f(\vect^{\wt(y_1)}, \ldots, \vect^{\wt(y_n)})$ is a non-zero $k$-variate polynomial of degree at most $dM$. Hence, the Schwartz-Zippel lemma would present a $(dM+1)^k$ sized hitting set. \\

\begin{definition}[Separating small sets of monomials]
  Let $S$ be an arbitrary set of monomials over $\vecy$. We shall say that a weight assignment $\wt:\vecy \rightarrow \N$ \emph{separates $S$} if for every distinct $m,m' \in S$ we have $\wt(m) \neq \wt(m')$.
\end{definition}

\begin{lemma}[\cite{AB03}] \label{lem:sparse-poly-pit}
  Let $S$ be an arbitrary set of $r$ multilinear monomials of degree at most $d$ over variables $\vecy = \setdef{y_{ij}}{i\in [d], j\in [n]}$. For a prime $p$, let $w_p:\vecy \rightarrow \N$ be a weight assignment given by
  \[
    w_p(y_{i,j}) = 2^{(i-1)n + (j-1)}\bmod{p}.
  \]
  Then for all but at most $\binom{r}{2} \cdot n^2$ primes $p$, the weight assignment $w_p$ separates $S$.
\end{lemma}

\subsubsection*{BIWAs for subspaces and products}

Agrawal, Gurjar, Korwar and Saxena \cite{AGKS15} constructed BIWAs for polynomials computed by ROABPs. The following two lemmas are slight abstractions of the key ideas in \cite{AGKS15}, so that they can also be applied in our setting. For the sake of completeness, the proofs are provided in \autoref{sec:biwa-subspace-products}. 

\begin{restatable}[BIWA for subspaces]{lemma}{BIWASubspace}
  \label{lem:biwa-of-subspaces}
  Say $V$ is a vector space of polynomials and suppose $\wt$ is a BIWA for $V$. Then, if $V'$ is a subspace of $V$, then $\wt$ is a BIWA for $V'$ as well. 
\end{restatable}

\begin{restatable}[BIWA for variable disjoint products]{lemma}{BIWAProducts}
  \label{lem:biwa-of-products}
  Say $V_1 \subseteq \F[\vecy]$ and $V_2 \subseteq \F[\vecz]$ are two vector spaces of polynomials over disjoint sets of variables, and of dimension at most $s$. Suppose
  \begin{align*}
    \wt_1&:\vecy\rightarrow \N^k\\
    \wt_2&:\vecz\rightarrow \N^k
  \end{align*}
  are BIWAs for $V_1$ and $V_2$ isolating bases $B_1$ and $B_2$ respectively. If
  $w: \vecy \union \vecz \rightarrow \N$
  is a weight assignment that \emph{separates} $B_1\cdot B_2 = \setdef{m_1m_2}{m_1\in B_1\,,\, m_2 \in B_2}$. Then the weight assignment defined by 
  \begin{align*}
    \wt&:\vecy\union \vecz\rightarrow \N^{k+1}\\
    \wt&: y_i  \mapsto (\wt_1(y_i), w(y_i)) \quad \text{for all $y_i \in \vecy$},\\
    \wt&: z_i  \mapsto (\wt_2(z_i), w(z_i)) \quad \text{for all $z_i \in \vecz$},
  \end{align*}
  is a BIWA for $V = V_1 \cdot V_2 = \Fspan\setdef{f\cdot g}{f\in V_1\,,\,g\in V_2}$. 
  
\end{restatable}

\subsection{Hitting sets for  UPT set-multilinear circuits}

\begin{restatable}[Hitting sets for UPT set-multilinear circuits]{theorem}{BIWAforUPT}
\label{thm:BIWAforUPT}
Let $\mathcal{C}$ be the class of $n$-variate degree $d$ set-multilinear polynomials (with respect to $\vecy = \vecy_1 \sqcup \cdots \sqcup \vecy_d$) that are computable by UPT set-multlinear circuits of preimage-width $w$ and depth $r$. Then, for $M = \inparen{\binom{w}{2}n^2 d + 1}^2$, the set
  \[
    \mathcal{H} = \setdef{(b_{11},\ldots, b_{dn})}{\vecp \in [M]^r\;,\; a_k \in A\;,\; b_{ij} = \prod_{k=1}^{r+1} a_k^{2^{(i-1)n + (j-1)}\bmod{p_i}} }
  \]
is a hitting set for $\mathcal{C}$ of size $\poly(ndw)^{r}$. 
\end{restatable}

The proof of this theorem is obtained by constructing what is called a \emph{basis isolating weight assignment} for polynomials simultaneously computed by a multi-output $\ComUPT$ circuit, heavily borrowing from the ideas in \cite{AGKS15}. 

\begin{proof}
  Suppose $f(\vecy)$ is a polynomial that is computable by a UPT set-multilinear circuit $C$ with respect to $\vecy = \vecy_1\sqcup \cdots \sqcup \vecy_d$ and say $C$ is of preimage-width size $w$ and depth $r$.

  Since $C$ is a UPT set-multilinear circuit, let $T$ be the shape of the parse tree. For each $\tau \in T$, we define the vector space
  \[
    V_{\tau} = \Fspan\setdef{[g]}{g\in C\;,\;g \sim \tau}. 
  \]
  The following claim relates the vector space corresponding to nodes in $T$ to the vector spaces corresponding to the children. 
  \indentBlock{
    \begin{claim}\label{claim:V-tensor-children}
      If $\tau \in T$ labels a $+$ gate and if $\tau'$ is the unique child of $\tau$, then $V_{\tau}\subseteq V_{\tau'}$.

      If $\tau \in T$ labels a $\times$ gate and has children $\tau_1$ and $\tau_2$, then $V_\tau$ is a subspace of $V_{\tau_1} \cdot V_{\tau_2}$.  
    \end{claim}
    \begin{nproof}{\autoref{claim:V-tensor-children}}
      Suppose $\tau \in T$ labels a $+$ gate and say $\tau'$ is the unique child of $\tau$ in $T$.
      Pick an arbitrary $g\in C$ such that $g \sim \tau$.
      If $[g] = [g_1] + \cdots + [g_s]$, then each $g_i \sim \tau'$.
      Therefore, $[g_i] \in V_{\tau'}$ and $[g] = [g_1] + \cdots + [g_s]$ implies that $[g] \in V_{\tau'}$. Since the choice of $g$ was an arbitrary gate of type $\tau$, it follows that $V_\tau$ is a subspace of $V_{\tau'}$.
      
      Say $\tau$ labels a $\times$ gate, and say $\tau_1$ and $\tau_2$ are the children of $\tau$. Pick an arbitrary gate $g\in C$ with $g \sim \tau$. If $[g] = [g_1]\times [g_2]$  then $g_1 \sim \tau_1$ and $g_2 \sim \tau_2$. But that implies that $[g_1] \in V_{\tau_1}$ and $[g_2] \in V_{\tau_2}$ and therefore $[g] \in V_{\tau_1} \cdot V_{\tau_2}$. Once again, since the choice of $g$ was arbitrary, we get $V_\tau$ is a subspace of $V_{\tau_1} \cdot V_{\tau_2}$. 
    \end{nproof}
  }
  
  Define the \emph{multiplication height} of any gate $g$, denoted by $\abs{g}_\times$, as the largest number of $\times$ gates encountered on a path from $g$ to a leaf. Starting with the leaves, we shall build towards a BIWA for $V_{\text{root}}$, which by \autoref{lem:biwa-preserve-nonzero} also yields a hitting set. 
  
  Let $P$ be the set of the first $(dn^2\binom{w}{2}+1)$ primes. For each $0\leq k\leq r$ and $\vecp = (p_1,\ldots, p_k)\in P^k$, define the function
  \begin{align*}
    \Omega_\vecp^{(k)}:\vecy & \rightarrow \N^{k+1}\\
    \Omega_\vecp^{(k)}: y_{ij}&\mapsto (j, 2^{(i-1)n+(j-1)}\bmod{p_1} ,\ldots, 2^{(i-1)n+(j-1)}\bmod{p_k}).
  \end{align*}
  The plan is to use $\Omega_\vecp^{(k)}$ to build BIWAs for each $V_\tau$. For a $\tau \in T$ with  $\abs{\tau}_\times = k$, let $S_\tau \subseteq [d]$ be the subset of indices labelling $\tau$. Define $\wt^{(\tau)}_\vecp$ to be the restriction of $\Omega_{\vecp}^{(k)}$ to $\union_{i \in S_{\tau}} \vecy_i$:
  \begin{align*}
    \wt^{(\tau)}_\vecp: \Union_{i\in S_\tau}\vecy_i &\rightarrow \N^{k+1}\\
    \wt^{(\tau)}_{\vecp}(y_{ij}) &= \Omega_\vecp^{(k)}(y_{ij}).
  \end{align*}
We shall prove, by induction, that for each $0\leq k\leq r$ there is a $\vecp\in P^k$ such that for every $\tau\in T$ with $\abs{\tau}_\times\leq k$, the weight assignment $\wt^{(k)}_\vecp$ is a BIWA for $V_{\tau}$. 
  
  If $\tau$ was a leaf of $T$, then any such node just computes a variable. Clearly, $\wt^{(\tau)}_\vecp:(y_{ij})\mapsto j$ is a BIWA as it gives distinct weights to all variables of a partition. Hence, $\wt^{(\tau)}_\vecp$ is a BIWA for all $V_{\tau}$ whenever $\tau$ is a leaf.
  
  If $\tau$ is not a leaf but $\abs{\tau}_\times = 0$, then neither $\tau$ nor its descendants are $\times$ gates. Hence, the subtree at $\tau$ has a unique leaf $\ell$ and all the nodes along this path are $+$ gates. By \autoref{claim:V-tensor-children}, $V_\tau$ is a subspace of $V_{\ell}$ and hence, by \autoref{lem:biwa-of-subspaces}, $\wt^{(\tau)}_\vecp = \wt^{(\ell)}_\vecp$ is a BIWA for $V_\tau$. That finishes the base case of $k = 0$.

  Suppose we have proved the claim up to $k-1$. Let $T_k$ be the set of all nodes of multiplication height at most $k$ that are $\times$ gates. By the inductive hypothesis, there exists $\vecp \in P^{k-1}$ such that $\wt^{(\tau')}_\vecp$ is  BIWA for all $V_{\tau'}$ with $\abs{\tau'}_\times < k$. Fix such a $\vecp$. For each $\tau \in T_k$, its children $\tau_1,\tau_2$ must have multiplication height at most $k-1$. Since $C$ is set-multilinear, the subset of indices that label $\tau_1$ and $\tau_2$ must be disjoint. Say $S_1$ and $S_2$ are the subsets of indices labelling $\tau_1$ and $\tau_2$ respectively.
  
  Hence, by \autoref{claim:V-tensor-children}, $V_\tau$ is a subspace of $V_{\tau_1} \cdot V_{\tau_2}$. By our inductive hypothesis, we know that $\wt^{(\tau_1)}_\vecp$ and $\wt^{(\tau_2)}_{\vecp}$ are BIWAs for $V_{\tau_1}$ and $V_{\tau_2}$ respectively. Observe that $\Omega^{(k-1)}_\vecp$ restricted to the appropriate subset of variables is a refinement of the weight assignments $\wt^{(\tau_1)}_\vecp$ and $\wt^{(\tau_2)}_\vecp$ (as $\abs{\tau_1}_\times$ or $\abs{\tau_2}_\times$ could have been smaller than $k-1$). Nevertheless, if $\wt^{(\tau_1)}_\vecp$ and $\wt^{(\tau_2)}_\vecp$ are BIWAs for $V_{\tau_1}$ and $V_{\tau_2}$ respectively, then  the following weight assignments
  \begin{align*}
    \wt_1 &: \Union_{i \in S_1}\vecy_i\rightarrow \N^k & \wt_2 &: \Union_{i\in S_2}\vecy_i\rightarrow \N^k\\
    \wt_1 &: y_{ij} \mapsto  \Omega^{(k-1)}_\vecp(y_{ij}) & \wt_2 &: y_{ij} \mapsto  \Omega^{(k-1)}_\vecp(y_{ij})
  \end{align*}
  are also BIWAs for $V_{\tau_1}$ and $V_{\tau_2}$ respectively. 
  By using \autoref{lem:biwa-of-products}, \autoref{lem:biwa-of-subspaces} and \autoref{lem:sparse-poly-pit}, besides perhaps $\binom{w}{2} n^2$ primes $p \in P$, the weight assignment defined by
  \begin{align*}
    \wt&:\Union_{i\in S_1 \cup S_2}\vecy_i\rightarrow \N^{k+1}\\
    \wt(y_{ij}) &= \begin{cases}
      (\wt_1(y_{ij}), 2^{(i-1)n+(j-1)}\bmod p) & \text{if $i \in S_1$},\\
      (\wt_2(y_{ij}), 2^{(i-1)n+(j-1)}\bmod p) & \text{if $i\in S_2$},
    \end{cases}\\
       &= (\Omega_\vecp^{(k-1)}(y_{ij}), 2^{in+j}\bmod p)
  \end{align*}
  is a BIWA for $V_\tau$. For different $\tau$s in $T_k$ there may a different set of $\binom{w}{2} n^2$ primes that we should exclude. But since the set $P$ of primes is at least $\binom{w}{2} n^2d + 1$, there is a prime $p \in P$ for which $\wt(y_{ij}) = (\Omega_\vecp^{(k-1)}, 2^{(i-1)n + (j-1)}\bmod p)$ is a BIWA for every $V_\tau$ where $\tau \in T_k$. By extending $\vecp$ by $p$ in the last coordinate, this shows that there is a $\vecp' \in P^k$ such that for each $\tau \in T_k$, the weight assignment $\wt_{\vecp'}^{(\tau)}$ is a BIWA for $V_\tau$.
  
  To complete the inductive step, we also need to prove the same for $\tau \in T$ that are $+$ gates with $\abs{\tau}_\times = k$. Hence, there must be a $\times$ gate $\tau' \in T_k$ that is a descendant of $\tau$ such that the path from $\tau$ to $\tau'$ consists only of $+$ gates. Once again, this forces $\wt^{(\tau)}_\vecp = \wt^{(\tau')}_\vecp$ and $V_\tau$ is a subspace of $V_{\tau'}$. Hence, by \autoref{claim:V-tensor-children} and \autoref{lem:biwa-of-subspaces}, it follows that $\wt^{(\tau)}_\vecp = \wt^{(\tau')}_\vecp$ is a BIWA for $V_\tau$ as well. And that completes the proof of the inductive step.\\
  
  Hence, if $f$ is a polynomial computed by a preimage-width $w$ UPT set-multilinear circuit of depth $r$, $\Omega_\vecp^{(r)}$ is a BIWA for $V_{\text{root}}$. Furthermore, by the prime number theorem, we know that the $\inparen{\binom{w}{2} n^2 d + 1}$-th prime cannot be bigger than $\inparen{\binom{w}{2} n^2 d + 1}^2$. Hence, the constructed BIWA is in fact a map
  \[
    \Omega_\vecp^{(r)}: \vecy \rightarrow [M]^{r+1}
  \]
  where $M \leq \inparen{\binom{w}{2} n^2 d + 1}^2$. Therefore, by \autoref{lem:biwa-preserve-nonzero} and the Schwartz-Zippel lemma, if we pick a set $A \subseteq \F$ with  $|A| > d \cdot \inparen{\binom{w}{2} n^2 d + 1}^2$, then
  \[
    \mathcal{H} = \setdef{(b_{11},\ldots, b_{dn})}{\vecp \in [M]^r\;,\; a_k \in A\;,\; b_{ij} = \prod_{k=1}^{r+1} a_k^{2^{(i-1)n + (j-1)}\bmod{p_i}} }
  \]
  is a hitting set for preimage-width $w$ depth $r$ UPT set-multilinear circuits and $\abs{\mathcal{H}} = \poly(ndw)^r$. 
\end{proof}

\subsection{Poly-sized hitting sets for constant width UPT circuits}

\ConstantWidthUPT*

\noindent The proof is an easy extension of the ideas from \cite{GKS16}, the details of which are in \autoref{subsec:constantWidthUPT-appendix}.

\section{FewPT circuits}

In this section we describe the black-box identity test for FewPT$(k)$ circuits. The following lemma from \cite{LLS17} shows that this class is equivalent to polynomials computed by sum of $k$ UPT circuits (of possibly different shapes).  

\subsection{Preliminaries}

\begin{lemma}{\rm (\cite[Lemma 16]{LLS17})}\label{lem:FewPT-to-sum-UPT}
  Let $f(\vecx)$ be a polynomial computed by FewPT$(k)$ circuit of preimage-width $w$. Then $f$ can be equivalently computed by a sum of $k$ UPT circuits of preimage-width $w$ each. 
\end{lemma}

Like in \cite{LLS17}, we'll refer to this class by $\SumUPT{k}$. We shall further qualify this notation to use $\SumUPT{k}(w)$ to denote the class of circuits that is a sum of $k$ UPT circuits of preimage-width $w$.

From this lemma, we can focus our attention on constructing hitting sets for $\ComSumUPT{k}$ circuits.  The proof largely follows the ideas of Gurjar, Korwar, Saxena and Thierauf \cite{GKST15}\footnote{\cite{GKST15} constructed hitting sets for sums of ROABPs and we use similar techniques for sums of UPT circuits. Roughly speaking, if we have a class $\mathcal{C}$ that has a \emph{characterizing set of dependencies} for which we know how to construct BIWAs, then we can also construct hitting sets for $\Sigma^k \mathcal{C}$.}.

\subsubsection*{Notation}

Let $\vecy = \vecy_1 \sqcup \cdots \sqcup \vecy_d$ be a partition of the variables and let $S = \inbrace{s_1,\ldots,s_p}$ be a subset of $[d]$.
Define the set of variables $\vecy_{S} = \vecy_{s_1} \union \cdots \union \vecy_{s_p}$ and the set of monomials $\vecy^S = \vecy_{s_1} \times \cdots \times \vecy_{s_p}$.
Also, define $\vecy_{-S} = \vecy \setminus \vecy_{S}$ and $\vecy^{-S} = \vecy^{[d] \setminus S}$.

\begin{definition}[Coefficient operator]
\label{defn:coefficientOperator}
Given a set-multilinear polynomial $f = \sum_{m \in \vecy^{[d]}} \alpha_m m$ of degree $d$, for $S \subseteq [d]$ and a monomial $m \in \vecy^{S}$, define $\Coeff{m} : \F\insquare{\vecy} \rightarrow \F\insquare{\vecy_{-S}}$ to be as follows.
\[
	\Coeff{m}(f) = \sum_{m' \in \vecy^{-S}} \alpha_{\inparen{m \cdot m'}} m'
\]
where $\alpha_{\inparen{m \cdot m'}}$ is the coefficient of $m m'$ in $f$.
\end{definition}

\begin{restatable}{lemma}{CommonUPTLemma}
  \label{lem:common-UPT-lemma}
  Let $\vecy = \vecy_1 \sqcup \ldots \sqcup \vecy_d$ be a partition and $f(\vecy)$ be a set-multilinear polynomial (with respect to the above partition) computed by a $\ComUPT$ circuit of preimage-width $w$ and underlying parse-tree shape $T$. Suppose $g(\vecy)$ is another set-multilinear polynomial (under the same partition) that \emph{cannot} be computed by a $\ComUPT$ circuit of preimage-width $w$ with the same shape $T$.

  Then, there exists $S \subseteq [d]$ and $R \in \F[\vecy_S]^{1 \times w'}$, and $P,Q \in \F[y_{-S}]^{w'\times 1}$ with $w'\leq w^2$ such that:
  \begin{itemize}\itemsep0pt
  \item For each $i \in [w']$, there is a monomial $m_i \in \vecy^{S}$ such that the $i$-th element of $P$ and $Q$ is $\coeff_{m_i}(f)$ and $\coeff_{m_i}(g)$ respectively,
  \item there is a vector $\Gamma \in \F^{1 \times w'}$ of support size at most $w+1$ such that $\Gamma P = 0$ and $\Gamma Q \neq 0$,
  \item the coefficient space of $R$ is full-rank, i.e. if we interpret $R$ as a matrix over $\F$ by listing each of its $w'$ entries as a column vector of coefficients, then this matrix has full column-rank.
  \item the vector of polynomials $R$ is simultaneously computable by a $\ComUPT$ circuit of preimage-width at most $w'$.
  \end{itemize}
  
\end{restatable}

This lemma is a fairly natural and straightforward generalization of \cite[Lemma 4.5]{GKST15} and a proof of this is provided in the appendix (\autoref{sec:FewPT-appendix}). 

\begin{lemma}\label{lem:FewPT-PIT-induction-step}
  Suppose $f(\vecy)$ is a non-zero polynomial computed by a $\ComSumUPT{k}(w)$  circuit. Suppose $\mathrm{wt}:\vecy \rightarrow M^r$ is a weight assignment that satisfies the following properties:
  \begin{itemize}\itemsep0pt
  \item $\wt$ is a BIWA for spaces of polynomials simultaneously computed by $\ComUPT$ circuits of preimage-width at most $w(w+1)$,
  \item For any $g$ in $\ComSumUPT{k-1}(w(w+1))$, the polynomial $g(\vecy + \vect^{\mathrm{wt}}) \in \F(\vect)[\vecy]$ has a monomial with non-zero coefficient that depends on at most $\ell$  distinct variables in $\vecy$.
  \end{itemize}
  Then, the polynomial $f(\vecy + \vect^{\mathrm{wt}})$ has a monomial, depending on at most $\log(w(w+1)) + \ell$ distinct variables in $\vecy$, with a non-zero coefficient. 
\end{lemma}

This is essentially a restatement of \cite[Lemma 4.6, Lemma 4.8]{GKST15} and follows from their proof. Unravelling the recursion, we get the following corollary. 

\begin{corollary}
  Let $f(\vecy)$ be a non-zero polynomial that can computed by a $\ComSumUPT{k}(w)$ circuit. Suppose $\mathrm{wt}:\vecy \rightarrow M^r$ is a BIWA for the class of polynomials simultaneously computed by $\ComUPT$ circuits of preimage-width at most $w^{2^{O(k)}}$. Then, the polynomial $f(\vecy + \vect^{\mathrm{wt}})\in \F(\vect)[\vecy]$ has a monomial with a non-zero coefficient that depends on at most $2^{O(k)}\log w$ variables in $\vecy$. 
\end{corollary}

Once we are guaranteed to retain a monomial of small-support, we can construct a hitting set by enumerating over all possible supports and applying the Schwartz-Zippel lemma \cite{O22,DL78,S80,Z79} (or apply standard generators such as the Shpilka-Volkovich generator \cite{SV15}). This completes the proof of \autoref{thm:hitting-set-few-parse-trees}, which we restate below for convenience. 

\FewPTPIT*

\section{Open problems}

An interesting open problem (at least to us) is whether we can give non-trivial hitting sets for the class of \emph{non-commutative skew circuits}. Lagarde, Limaye and Srinivasan \cite{LLS17} provide a white-box PIT in some restricted settings when the skew circuits are somewhat closer to UPT (with some restriction on what sort of parse trees they can have) but removing this restriction would be a great step forward.

Another issue is that the current construction of hitting sets for FewPT circuits (which build on \cite{GKST15}) incurs quasipolynomial losses at two different places. The first is in the construction of the \emph{basis isolating weight assignment (BIWA)}, and we only know to construct that using quasipolynomially large weights. The other is in a brute-force enumeration of all monomials of support $O(\log s)$. As a result, even if at a later day we have a construction of a BIWA with polynomially large weights, this proof would  still only yield a quasipolynomially large hitting set for FewPT circuits. It would be interesting to see if this brute-force enumeration could be circumvented. 

\section*{Acknowledgements}

We thank the organizers of the \href{http://www.imsc.res.in/~meena/nmi17-wac/nmi-2017-arithmetic-complexity.html}{NMI Workshop on Arithmetic Complexity 2017} where we learned of the circuit classes that we study in this paper. 
We thank Nutan Limaye and Srikanth Srinivasan for numerous discussions that eventually led to these results. We thank Rohit Gurjar for pointing out a subtlety in a previous draft of this paper, and also thank Amir Shpilka for inviting RS to Tel Aviv University (where this discussion took place). 
  
\bibliographystyle{alphaurlpp}
\bibliography{references}

\appendix

\section{Separating ABPs from UPT circuits}
\label{appsect:abp-upt-sep}

This section contains the proofs of the separation between ABPs and UPT circuits. Recall the definition of the polynomial $P_d$ (of degree $D = 2^{d+1}-1$).
\[
P_d(x_1,\ldots, x_m) = \sum_{\substack{\gamma \in [m]^D\\\text{$\gamma$ is legal}}} x_{\gamma(v_1)} x_{\gamma(v_2)} \cdots x_{\gamma(v_D)}. 
\]

\subsubsection*{Upper bound}

\UPTABPSepUB*

\begin{proof}
Let $\mathcal{G}(d,\alpha)$ be the set of all legal colourings $\gamma$ with $v_{2^d}$ (root of $T_d$) satisfying $\gamma(v_{2^d}) = \alpha$. Now we define $P_{d,\alpha}(x_1,\ldots,x_m)$ as
\begin{align*}
P_{d,\alpha}(x_1,\ldots,x_m) = \sum_{\gamma \in \mathcal{G}(d,\alpha)} x_{\gamma(v_1)} x_{\gamma(v_2)} \cdots x_{\gamma(v_{D})}.
\end{align*}
Clearly, $P_d(x_1,\ldots,x_m) = \sum_{\alpha \in [m]} P_{d,\alpha}(x_1,\ldots,x_m)$. Therefore we can now recursively write
\begin{equation}\label{eqn:P_d-recDefn}
P_d(x_1,\ldots,x_m) = \sum_{\alpha,\beta \in [m]} P_{d-1,\alpha}(x_1,\ldots,x_m) \cdot x_{\alpha +_m \beta} \cdot P_{d-1,\beta}(x_1,\ldots,x_m),
\end{equation}
where $\alpha +_{m} \beta = (\alpha + \beta) \bmod m$.

Now using \eqref{eqn:P_d-recDefn} it is easy to see that if we have UPT circuits for $P_{d-1,\alpha}(x_1,\ldots,x_m)$s then a UPT circuit computing $P_{d}(x_1,\ldots,x_m)$ can be obtained and this follows directly by induction. Hence, repeated application of \eqref{eqn:P_d-recDefn} yields a UPT circuit computing $P_d$ of size $O(m^2 d)$. 
\end{proof}

\subsubsection*{Lower bound}

As mentioned earlier, much of the lower bound argument is exactly along the lines of the proof of \cite{HY16}. The modifications required from their proof are quite minor but we present the proof here for completeness. 

\UPTABPSepLB*

\begin{proof}
  Let us fix some $\sigma \in S_D$ and let $Q(x_1,\ldots, x_m) = \Delta_\sigma(P_d)$. In order to show that $Q$ requires ABPs of large width, it suffices to show that there exists some $0 \leq k \leq D$ for which the partial derivative matrix, given by

  \begin{tikzpicture}
    \draw[fill=black!5] (0,0) rectangle (5,3);
    \node at (-3.5,1.5) {$M_k(Q) = $};
    
    \draw[decorate,decoration={brace,amplitude=10pt,raise=4pt},yshift=0pt]
    (0,0) -- (0,3);
    \node[anchor=east] at (-0.5,1.5) {$[m]^k$};
    \draw[decorate,decoration={brace,amplitude=10pt,mirror, raise=4pt},yshift=0pt] 
    (0,0) -- (5,0);
    \node[anchor=north] at (2.5,-0.5) {$[m]^{D-k}$};
    
    \draw[fill=black!10] (1,0) rectangle (1.5,3);
    \node at (1.25,3.2) {$w$};
    
    \draw[fill=black!10] (0,2) rectangle (5,2.5);
    \node[anchor=west] at (5.2,2.2) {$w'$};
    
    \draw[very thick] (1,2) rectangle (1.5,2.5);
    
    \node[anchor=west] at (6,1) {coefficient of $x_w\cdot x_{w'}$ in $Q$}
    edge[<-,bend left] (1.25,2);
  \end{tikzpicture}

  \noindent
  has rank at least $m^{\Omega(d)}$. We shall prove this by exhibiting an  $r\times r$ identity matrix as a submatrix in $M_k(Q)$ with $r = m^{\Omega(d)}$. The $k$ that we will work with would be the number whose binary expansion is $10101\cdots$. The relevance for this comes from the fact that the \emph{edge boundary} of any subset $V_0 \subseteq T_d$ is with $|V_0| =  k$ for such a $k$ is reasonably large. 

  \begin{definition}[Isoperimetric profile of graphs]\label{defn:eip}
    Given a graph $G = \inparen{V(G),E(G)}$ and a subset of vertices $A \subseteq V(G)$, edge isoperimetric profile of $G$ is given by the following function $\eip(k)$ defined by 
    \begin{equation*}
      \eip_G(k) = \min \setdef{\abs{E(A,\overline{A})}}{A \subseteq V(G), \abs{A}=k},
    \end{equation*}
    where $E(A,\overline{A})$ is the set of edges with one end-point in $A$ and the other outside. 
  \end{definition}

  \begin{lemma}{\cite{HY16}}\label{lem:hy-good-k}
    If $k \leq D$ is the number whose binary expansion is $1010\cdots$, then $\eip_{T_d}(k) \geq \frac{d}{4}$. 
  \end{lemma}

  The relevance for this would become apparent shortly, but let us proceed for now. If there is indeed an ABP for a shuffling of $f$, then the rows of $M_k(Q)$ is just a partial colouring of a subset $V_0 \subset T_d$ of size exactly $k$. Similarly, the columns of $M_k(Q)$ are partial colourings of $V_1 := T_d \setminus V_0$.  Therefore $M_k(Q)_{\inparen{x_w,x_{w'}}}$ is $1$ only if the colouring of $V_0$ given by $x_w$ and that of $V_1$ given by $x_{w'}$ together form a legal colouring of $T_d$. Hence the task of finding an $r \times r$ submatrix of $M_k(Q)$ reduces to finding colourings $C_1,C_2, \ldots, C_r$ of $V_0$ and colourings $C'_1, C'_2, \ldots, C'_r$ of $V_1$ such that the colouring $C_i \circ C'_j$ is legal if and only if $i=j$, for all $i,j \in [r]$.

We will need the notion of \emph{pure nodes} (as defined by \cite{HY16}).
  
  \begin{definition}{(\emph{Pure} nodes).}\label{defn:pureNode}
    For $i \in \inbrace{0,1}$, a non-leaf node $v$ in $V_i$ is called said to be \emph{pure} if there is a path $\Pi = \inparen{v,v_1,v_2,\ldots,v_k}$ in $T_d$ where $v_k$ is a leaf that is a descendant of $v$, and $\Pi \cap V_i = \inbrace{v}$.
  \end{definition}
  There may be multiple witnesses $v_k$ for the fact that $v$ is a pure node. For each pure node, we shall assign one leaf arbitrarily as its \emph{pure leaf}. It is easy to see that the pure leaves are distinct for each pure node. 
  
  Let the pure nodes in $V_0$ be $P_0$ and those in $V_1$ be $P_1$ and say $P:= P_0 \union P_1$. Let $\ell(P)$, $\ell(P_0)$ and $\ell(P_1)$ be the pure leaves of $P$, $P_0$ and $P_1$ respectively. 
  \begin{lemma}{\rm (\cite[Claim 11]{HY16})}
    $\abs{P} \geq \frac{\abs{E(V_0,V_1)}}{4}.$
  \end{lemma}
  Without loss of generality, we may assume that $P_0$ is bigger than $P_1$ and the above lemma, in conjunction with \autoref{lem:hy-good-k}, gives that $|P_0| \geq d/32$. We are now ready to define our colourings $C_1,\ldots, C_r$ and $C_1',\ldots, C_r'$ for $r =  m^{|P_0|} \geq m^{d/32}$. 

Let $L$ be the set of all leaves in $T_d$. For each $\vecc_i \in [m]^{|P_0|}$, define $\tilde{C_i}:T_d \rightarrow \Z_m$ obtained by assigning colour $1$ to all leaves in $L \setminus \ell(P_0)$,  assigning $\vecc_i$ to the leaves in $\ell(P_0)$ and extending it uniquely to the other vertices of $T_d$ in order to make it legal. The partial colourings $C_i$ and $C_i'$ be the restriction of $\tilde{C_i}$ to $V_0$ and $V_1$ respectively. 

Clearly, $C_i \circ C_i' = \tilde{C_i}$ and hence is a valid colouring. Now consider $C_i$ and $C_j'$ for $i\neq j$. There must exist some leaf $v\in \ell(P_0)$ that gets different colours in $C_i$ and $C_j$ and let $u$ be the node in $P_1$ that $v$ was a pure leaf of. We shall assume that $u$ is \emph{minimal} in the sense that any pure node $u' \in P_1$ that is a descendant has all its leaves identically coloured in $C_i$ and $C_j$. But then, the colour of $u$ in $\tilde{C_i}$ and in $\tilde{C_j}$ cannot be the same as exactly one leaf if $u$ has a different colour in $\tilde{C_i}$ and $\tilde{C_j}$ respectively. This would then imply that $C_i$ forces $u$ to be given a colour different than what $C_j'$ assigns and hence $C_i \circ C_j'$ is not legal. 

Therefore, this shows that the matrix $M_k(Q)$ has an $r\times r$ identity submatrix with $r \geq m^{d/32}$. Therefore, any ABP computing $Q$ must have width at least $m^{\Omega(d)}$. 
\end{proof}

\section{Exponential lower bound under any shuffling}
\label{appsect:exp-under-shuffling}

Here we give an explicit polynomial that has polynomial sized arithmetic circuits but requires exponential sized UPT circuits under any shuffling. A version of the hard polynomial appears in \cite{LMP16}. They show that the polynomial requires exponential sized UPT circuits and that it is efficiently computable by what are known as \emph{skew circuits} (see \cite{LMP16} for a formal definition). Here we extend the lower bound and show that it applies to any \emph{shuffling} of the polynomial.

\subsection{The polynomial}

The hard polynomial we discuss is called the \emph{moving palindrome} which is a variant of the \emph{palindrome} polynomial. The palindrome polynomial of degree $d$ on $n$ variables, as known, is defined as follows.
\begin{equation*}
\Pal_d(x_1,\ldots,x_n) := \sum_{ w \in \inbrace{x_1,\ldots,x_n}^{d/2}} w \cdot w^R
\end{equation*}
where $w^R$ denotes the reverse of the word $w$.

\noindent 
Using this definition, we define the $(n+1)$-variate moving palindrome of degree $D$ as follows.
\begin{equation*}
\MPal_{D}(x_1,\ldots,x_n,z) := \sum_{ 0 \leq \ell \leq D/2 } z^{\ell} \cdot \Pal_{\frac{D}{2}}(x_1,\ldots,x_n) \cdot z^{\frac{D}{2} - \ell}
\end{equation*}

\subsection{The lower bound}

Similar to the matrix $M_k$ defined in \autoref{appsect:abp-upt-sep} for a commutative polynomial, define a \emph{partial derivative} matrix $M_{(i,p)}$ for a non-commutative polynomial $g$. Here the $(w,w')$ entry of $M_{(i,p)}$ will be the coefficient of $w \times_p w'$ in $g$, where $\deg(w) = i$. We will show that $M_{(i,p)}$ for $\MPal_D$ has rank $n^{\Omega(D)}$ for a \emph{range of types} $(i,p)$, such that any UPT circuit computing any shuffling of $\MPal_{D}$ must admit at least one of those types. Then using the characterization from \cite{LMP16}, we will conclude the following theorem.

\begin{theorem}\label{thm:expLowerBound}
For any $\sigma \in S_D$, a UPT circuit computing $\Delta_{\sigma}(\MPal_{D})$ has $n^{\Omega(D)}$ gates.
\end{theorem}

\begin{proof}
  Let $2d$ be the degree of the palindrome, giving $D=4d$. Also, let $P_{\ell}(\vecx,z) = z^{\ell} \Pal_{2d}(\vecx) z^{2d-\ell}$. Therefore $\MPal_{D} = \sum_{\ell=0}^{2d} P_{\ell}(\vecx,z) = f(\vecx,z)$ (say). For $P_{\ell}$, and for $\ell < j_1, j_2 \leq 4d-\ell$, we will say that $j_1$ and $j_2$ are \emph{dependent} with respect to $P_{\ell}$ if all monomials in $P_{\ell}$ contain the same variable in positions $j_1$ and $j_2$. It is easy to see that the criterion $j_1 + j_2 = 2(d+\ell)+1$ captures this relation. Define a \emph{dependency} graph $G_{\ell}=(V,E_{\ell})$ with $V=\{1,2,\ldots,4d\}$ such that $(j_1,j_2) \in E_{\ell}$ if and only if $j_1$ and $j_2$ are dependent with respect to $P_{\ell}$. Let $G=(V,E)$ with $E = \cup_{\ell} E_{\ell}$.

  If $[4d] = V_0 \sqcup V_1$ is a partition, let us define a matrix $\tilde{M}_{V_0,V_1}(f)$ to be the one where rows and columns are indexed by a partial assignment to the positions $V_0$ and $V_1$ respectively. 

  \begin{claim}\label{claim:cutGivesLB}
    Let $[4d] = V_0 \sqcup V_1$ be a partition of the positions, and suppose that for some $\ell \in \set{0,\ldots, 2d}$ we have $t$ edges in $E_\ell$ crossing the cut $(V_0,V_1)$ in $G_\ell$. Then, $\rank \inparen{\tilde{M}_{V_0,V_1}(f)} \geq n^t$. 
\end{claim}
\begin{proof}
  In the polynomial $P_\ell$, let $Z_\ell \subseteq [4d]$ be the positions that are fixed to $z$. Consider the submatrix of $\tilde{M}_{V_0,V_1}$ where $V_0 \intersection Z_\ell$ and $V_1 \intersection Z_\ell$ are assigned to $z$. Observe that this submatrix is precisely $\tilde{M}_{V_0', V_1'}(P_\ell)$ where $V_0' = V_0 \intersection \overline{Z_\ell}$ and $V_1' = V_1 \intersection \overline{Z_\ell}$. 

If we have $t$ edges crossing the cut $(V_0',V_1')$ (none of the cut edges can be adjacent on $Z_\ell$), then we have a size $t$ matching in $(V_0',V_1')$. This means that fixing the variables in their $V_0'$ end-points uniquely fixes their $V_1'$ end-points. Hence, it is clear that we have an $n^t\times n^t$ identity submatrix and hence that the rank of $\tilde{M}_{V_0,V_1}(f)$ is at least $n^t$. 
\end{proof}

The next claim shows that for any $V_0$ in a fairly wide range of sizes, there will always be some $\ell$ with $G_\ell$ exhibiting a large cut. 

\begin{claim}\label{claim:cutExists}
For any set $V_0 \subseteq [4d]$ of size $k$ with $\frac{d}{6} \leq k \leq \frac{d}{3}$, there is some $\ell \in \inbrace{0,\ldots,2d}$ such that $\Omega(d)$ edges in $E_\ell$ cross the cut $(V_0,V_1)$.
\end{claim}
\begin{proof}
Let $V_0$ be a set of $k$ positions with $k \leq \frac{d}{3}$. Let us partition the set of positions $V = \inbrace{1,2,\ldots,4d}$ into $S_1 = \inbrace{1,\ldots,k}$, $M_1 = \inbrace{(k+1),\ldots,(2d-k)}$, $T_1 = \inbrace{(2d-k+1),\ldots,2d}$, $T_2 = \inbrace{(2d+1),\ldots,(2d+k)}$, $M_2 = \inbrace{(2d+k+1),\ldots,(4d-k)}$ and $S_2 = \inbrace{(4d-k+1),\ldots,4d}$.

Now the possible choices for $V_0$ can be split into the following (possibly overlapping) cases:
\begin{enumerate}
	\item $V_0 \cap T_1 \geq \frac{k}{8}$:\\
	Note that the degree of any vertex in $T_1$ is at least $(2d-k)$, and that every even (or odd) vertex in $M_1$ is connected to every odd (or even) vertex in $S_2$. Now $V_1 \cap M_1$ is at least $2d - k - (k - \frac{k}{8}) \geq 2(d - k)$. Total number of edges crossing $(V_0,V_1)$ is therefore $\geq \abs{(V_0 \cap T_1, V_1 \cap M_1)} \geq 2 \inparen{\frac{1}{4} \times (d-k) \times \frac{k}{8}} = \Omega(dk)$. Therefore there exists an $E_i$ that achieves the average $\Omega(k) = \Omega(d)$ edges crossing the cut $(V_0,V_1)$.
	\item $V_0 \cap S_1 \geq \frac{k}{4}$:\\
	Consider the neighbourhood of $V_0 \cup S_1$ due to $E_0$. All these positions are in $T_1$. If more than $\frac{k}{8}$ of them are in $V_0$ then case 1 applies. Else we get that $\geq \frac{k}{8}$ edges from $E_0$ cross $(V_0,V_1)$.
	\item $V_0 \cap M_1 \geq \frac{k}{4}$:\\
	Again, every even (or odd) position in $M_1$ is connected to every odd (or even) position in $T_1$, the degree of every position in $M_1$ is at least $k$, and $\abs{V_1 \cap T_1} \geq \frac{k}{8}$. Therefore a total of $\Omega(k^2)$ edges cross $(V_0,V_1)$, thereby again giving us that some $E_i$ achieves $\Omega(d)$ edges crossing $(V_0,V_1)$.
\end{enumerate}
Since the other cases (with $T_2, S_2, M_2$) are symmetric to those discussed above, we can conclude the statement of the claim.
\end{proof}

In order to complete the proof, we just need to show that any UPT circuit computing a homogeneous degree $d$ polynomial, there will be a gate of position-type $(i,p)$ with $\frac{d}{6} \leq i \leq \frac{d}{3}$. 

\begin{lemma}\label{lem:kto2kGateExists}
For all $0 < \alpha < \frac{1}{2}$, any UPT circuit (with fan-in 2 $\times$ gates) computing a polynomial of degree $D$ contains a gate computing a degree $i$ polynomial for some $\alpha D \leq i \leq 2 \alpha D$.
\end{lemma}
\begin{proof-sketch}
Let $C$ be a UPT circuit computing a degree $D$ polynomial with multiplication gates of fan-in 2. Starting from the root of $C$, choose an arbitrary child at every addition gate and the child computing a higher degree polynomial at every multiplication gate. As the degree never drops to a fraction less than half in any step, we eventually reach an appropriate gate.
\end{proof-sketch}

Now \autoref{lem:kto2kGateExists} tells us that for any UPT circuit computing $\Delta_\sigma(\MPal_D)$, will have a gate of position-type $(i,p)$ with $\frac{D}{24} \leq i \leq \frac{D}{12}$. We can then apply \autoref{claim:cutExists} and then \autoref{claim:cutGivesLB} to obtain an $n^{\Omega(D)}$ lower bound on the number of gates in $C$.
\end{proof}

\section{Hitting sets for UPT circuits}
\label{appsec:UPT-PIT}

\subsection{Commutative analogue of UPT circuits}
\label{appsubsect:defn-UPT-SML}

Consider substitution map $\Phi : \set{x_1,\ldots, x_n} \rightarrow \F[y_{1,1},\ldots, y_{d,n}]^{(d+1) \times (d+1)}$ given by
  \[
    \Phi(x_i) = 
    \begin{bmatrix}
      0 & y_{1,i} & 0 & \ldots & 0 & 0 \\
      0 & 0 & y_{2,i} & \ldots & 0 & 0 \\
      0 & 0 & 0 & \ldots & 0 & 0 \\
      \vdots & \vdots & \vdots & \ddots & \vdots & \vdots\\
      0 & 0 & 0 & \ldots & 0 & y_{d,i}\\
      0 & 0 & 0 & \ldots & 0 & 0 \\
    \end{bmatrix},\quad\text{for all $i \in [n]$.}
  \]
  To understand the effect of $\Phi$ on a homogeneous non-commutative polynomial $f(x_1,\ldots, x_n)$ of degree $d$, define $\Psi:\F\inangle{x_1,\ldots, x_n}_{\deg=d} \rightarrow \F[y_{1,1},\ldots, y_{d,n}]$ as the unique $\F$-linear map given by

\noindent $\Psi: x_{w_1}\cdots x_{w_d} \mapsto y_{1,w_1}\cdots y_{d,w_d}$.

\medskip
  \begin{lemmawp}[\cite{FS13}]\label{lem:noncomm-to-comm}
    Let $f  = \sum_w a_w x_w \in \F\inangle{x_1,\ldots, x_n}$ be a homogeneous degree $d$ non-commutative polynomial. Then, $f$ under the substitution map $\Phi$ (defined above) is given by
    \[
      f \circ \Phi = f(\Phi(x_1),\ldots, \Phi(x_n)) =
      \begin{bmatrix}
        0 & \cdots & 0 & \Psi(f)\\
        0 & \cdots & 0 & 0\\
        \vdots & \ddots & \vdots & \vdots\\
        0 & 0 & 0 & 0
      \end{bmatrix}_{(d+1)\times (d+1)}\qedhere
    \]
  \end{lemmawp}

  Similar to the above definition of $\Psi$, we define a \emph{shifted} version of it called $\Psi_a$ (for a parameter $a \in \N$) as $\Psi_a: x_{w_1}\cdots x_{w_d} \mapsto y_{a+1,w_1}\cdots y_{a+d,w_d}$.

  \begin{observation}
    If $f\in \F\inangle{x_1,\ldots, x_n}_{\deg = d_1}$ and $g\in \F\inangle{x_1,\ldots, x_n}_{\deg = d_2}$, then for any $a \in \N$, we have $\Psi_a(f\cdot g) = \Psi_a(f) \cdot \Psi_{a+d_1}(g)$. 
  \end{observation}
  

  In the case of \cite{FS13}, when $f$ was computable by non-commutative ABPs, they showed that $\Psi(f)$ is computable by an ROABP. In our setting of non-commutative UPT circuits, the following is the commutative analogue.

  \begin{observation}
    Let $C$ be a UPT circuit computing a polynomial $f \in \F\inangle{x_1,\ldots, x_n}$ of size $s$ and depth $r$. Consider the commutative circuit $C'$ where each leaf variable of type $(1,p)$ that is labelled by $x_i$ is replaced by $y_{p+1,i}$. Then the circuit $C'$ computes $\Psi(f)$ and is UPT and set-multilinear with respect to $\vecy = \vecy_1 \sqcup \cdots \sqcup \vecy_d$ where $\vecy_i = \setdef{y_{i,j}}{j\in [n]}$. 
  \end{observation}

\subsubsection*{BIWAs for subspaces and products}
\label{sec:biwa-subspace-products}

\BIWASubspace*

\begin{proof}
  If $B$ is a monomial basis of $V$ that is isolated by $\wt$, then the columns indexed by $B$ span the column space of $V'$ as well. Starting with the columns of $V'$ indexed by $B$, pick a \emph{minimum weight basis} $B'$ according to $\wt$, so that any column of $V'$ that is outside $B'$ is spanned by lower weight monomials in $B'$. By definition $\wt$ is a BIWA of $V'$ isolating $B'$, as all columns in $B'$ get distinct weights and every column outside $B'$ is spanned by lower weight columns in $B'$. 
\end{proof}

\BIWAProducts*

\begin{proof}
  Observe that by the definition of $\wt$, $\wt(m_1\cdot m_2) = (\wt_1(m_1)  + \wt_2(m_2), w(m_1\cdot m_2))$ for any $m\in \Mons(\vecy)$ and $m' \in \Mons(\vecz)$. 
  
  If $V_1$ and $V_2$ are expressed as matrices (with the generators listed as rows), then the matrix corresponding to $V$ is just $V_1 \otimes V_2$, the tensor product.
Let $B_1 = \set{m_1,\ldots, m_r}$ and $B_2 = \set{m_1',\ldots, m_s'}$.
We shall prove that the weight assignment $\wt$ is a BIWA that isolates the natural spanning set $B = B_1 \cdot B_2 = \setdef{m_i m_j'}{i\in [r]\;,\; j\in [s]}$.
Firstly, note that all the elements of $B$ have distinct weights due to the presence of the last coordinate from $\wt$, which separates the $rs$ monomials in $B_1\cdot B_2$.

Now suppose $\tilde{m} = m\cdot m' \notin B$ for $m\in \Mons(\vecy)$ and $m' \in \Mons(\vecz)$ and say without loss of generality $m \notin B_1$. The column indexed by $\tilde{m}$ in $V_1 \cdot V_2$ is just the tensor product of the columns indexed by $m$ in $V_1$ and the column indexed by $m'$ in $V_2$. But since $\wt_1$ is basis isolating for $V_1$, the column of $V_1$ indexed by $m$ can be expressed as a linear combination of lower weight terms.
\begin{align*}
  V_{1,m} &= \sum_{\wt_1(m_i) \prec \wt_1(m)} a_i \cdot  V_{1,m_i}\\
  \implies V_{\tilde{m}} = V_{1,m} \otimes V_{2,m'} & =  \sum_{\wt_1(m_i) \prec \wt_1(m)} a_i \cdot  \inparen{V_{1,m_i} \otimes V_{2,m'}}\\
          & = \sum_{\wt_1(m_i) \prec \wt_1(m)} a_i \cdot  V_{m_im'}
\end{align*}
But notice that $\wt_1(m_i) \prec \wt_1(m)$ also implies that $\wt(m_im') \prec \wt(mm')$.
Therefore,  (repeating this argument on $m'$ if $m' \notin B_2$) we can write any column with index outside $B$ as a linear combination of columns of smaller weight in $B$. Hence, $\wt$ is indeed a BIWA for $V$ that isolates $B$. 
\end{proof}

\subsection{Constant width UPT circuits}
\label{subsec:constantWidthUPT-appendix}

In this subsection we prove the existence of a $\poly(n,d)$ hitting set for UPT circuits of constant preimage-width computing $n$-variate degree-$d$ polynomials, when the \emph{shape} of the circuit is known. The proof is an easy extension of the ideas of \cite{GKS16} to the $\ComUPT$ circuits regime. We will construct a univariate substitution map that preserves its nonzero-ness and has degree $\poly(n,d)$, which will imply a hitting set naturally. 

Say $\vecy = \vecy_1 \sqcup \cdots \sqcup \vecy_d$ and let $f(\vecy)$ be an $nd$-variate degree $d$ polynomial computable by a $\ComUPT$ circuit (with respect to the above partition) of constant preimage-width. From \autoref{obs:depth-reduction-width}, we may assume that the circuit has depth $\log{d}$. We will need the following lemma for bivariates over \emph{large} fields.

\begin{lemma}{\rm (\cite[Lemma 3.2]{GKS16})}
\label{lem:bivariateGKS}
Let $f(y_1,y_2) = \sum_{i=1}^w u_i(y_1)v_i(y_2)$ be a nonzero bivariate polynomial of degree $d$ over $\F$. If $\operatorname{char}(\F)=0$ or $\operatorname{char}(\F) > d$, then $f(t^w, t^{w-1} + t^w) \neq 0$.
\end{lemma}

Suppose $f(\vecy)$ is computable by a circuit $C$ that has shape $T$. Define the set of variables $\vect = \inbrace{t_{\tau} : \tau \in T}$. We will begin by substituting $t_{\tau_i}^j$ for every $y_{ij}$ where the leaf in $C$ computing polynomials over $\vecy_i$ corresponds to $\tau_i$ in $T$. As long as we can, we will pick a multiplication gate $\tau$ that has its left and right children (say $\tau_L$ and $\tau_R$) computing univariates in $t_{\tau_L}$ and $t_{\tau_R}$ respectively; and then substitute $t_{\tau_L} \gets t_\tau^w$ and $t_{\tau_R} \gets t_\tau^w + t_\tau^{w-1}$. Let us call this substitution $\Phi_{\tau}$.

\begin{lemma}
  \label{lem:iterationCorrectness}
  Consider the above iterative process of substituting some of the $\vecy_i$'s by suitable polynomials in $\vect$. Let $\tilde{\Phi}(f) = \tilde{f}(\vect, \vecy)\neq 0$ be the polynomial just before applying the substitution $\Phi_{\tau}$. Then $\tilde{f}' = \Phi_\tau(\tilde{f}) := \tilde{f}(t_{\tau_L} \gets t_{\tau}^w, t_{\tau_R} \gets t_{\tau}^w + t_{\tau}^{w-1}) \neq 0$. 
\end{lemma}
\begin{proof}
  From \eqref{eqn:vsbr-u}, we have
  \begin{align*}
    f &= \sum_{u \sim \tau} [u] \cdot [\mathrm{root}:u]\\
      & = \sum_{u \sim \tau} [u_L] \cdot [u_R] \cdot [\mathrm{root}:u],\\
    \implies \tilde{\Phi}(f) &= \sum_{u\sim \tau} a_u(t_{\tau_L}) \cdot b_u(t_{\tau_R}) \cdot h_u(\vecy, \vect \setminus {t_{\tau_L},t_{\tau_R}}) \neq 0.
  \end{align*}
  We may treat $\tilde{\Phi}(f)$ as a bivariate polynomial in $t_{\tau_L},t_{\tau_R}$ over the field $\F(\vect \setminus \inbrace{t_{\tau_L}, t_{\tau_R}})$ and apply \autoref{lem:bivariateGKS} to conclude that $\Phi_\tau(\tilde{\Phi}(f))$ will be nonzero if and only if $\tilde{\Phi}(f)$ was nonzero.
\end{proof}

Now for every leaf node in $T$, create a sequence which we will call its \emph{signature}, by walking down from the root to the leaf. Every time we pick the left child, we append $L$ to the signature and every time we pick the right child, we append $R$. For $\tau \in T$, call the sequence $\operatorname{sig}_{\tau} = \inparen{a_1~a_2~\cdots~a_r}$. Let $t$ be a fresh variable and $\tau_i$ be the node corresponding to $\vecy_i$. Define
\begin{align*}
  \Phi_L&: t \mapsto t^w\quad,\quad  \Phi_R: t \mapsto t^w + t^{w-1}\\
  \Psi &: \vecy \rightarrow \F[t]\\
  \Psi &: y_{ij} \mapsto \Phi_{a_1} \circ \Phi_{a_1} \circ \cdots \circ \Phi_{a_r}(t^j)
\end{align*}
where $\inparen{a_1 \cdots a_r} = \operatorname{sig}_{\tau_i}$.
Observe that the procedure described above essentially executes the substitution $\Psi$ on $\vecy$. We can then infer from \autoref{lem:iterationCorrectness} that for any $f(\vecy)$ computable by $\ComUPT$ circuits, $f(\vecy) \neq 0 \iff f(\Psi(\vecy)) \neq 0$. This gives us the following theorem.

\begin{theorem}
Let $f(\vecy)$ be a polynomial computed by an $\ComUPT$ circuit of width $w$ and depth $r$. Consider the following substitution $\Psi:\vecy \rightarrow \F[t]$ given by
\[
    \Psi: y_{ij}\mapsto  \Phi_{a_1}\circ \Phi_{a_2} \circ \cdots \circ \Phi_{a_r}(t^j),
\]
where the \emph{signature} of the part $\vecy_i$ is $a_1a_2\cdots a_r$. Then $f(\vecy)$ is non-zero if and only if $f(\Psi(\vecy))$ is non-zero. 
\end{theorem}

Now since the depth of the circuit is at most $O(\log{d})$, if the width is constant, then the final degree of $f(\Psi(\vecy))$ is at most $O(nw^{O(\log{d})})$, which is $\poly(n,d)$ if $w = O(1)$. This finishes the proof of \autoref{thm:hitting-set-low-width-known-shape}.

\section{Hitting sets for FewPT circuits}
\label{sec:FewPT-appendix}

We will need the following fact about coefficient operators (defined in \autoref{defn:coefficientOperator}).

\begin{observation}[Coefficients of UPT circuits are also UPT circuits]
\label{obs:UPTSMLforCoeff}
Suppose $f(\vecy)$ is a homogeneous degree $d$ polynomial that is computable by a UPT set-multilinear circuit with respect to $\vecy = \vecy_1 \sqcup \cdots \sqcup \vecy_d$ of preimage-width $w$. If $S \subseteq [d]$ and $m$ is any monomial in $\vecy^{S}$, then the polynomial $\Coeff{m}(f)$ can also be computed by a UPT set-multilinear circuit of preimage-width $w$. 
\end{observation}
\begin{proof-sketch}
Since the $\ComUPT$ circuit $C$ can be made canonical without loss of generality, we only need to set the corresponding leaves in $\vecy_S$ as $0$ or $1$ depending on whether the variable appears in $m$. 
\end{proof-sketch}

\noindent
The following is an analogue of \cite[Lemma 4.5]{GKST15}. 

\CommonUPTLemma*

\begin{proof}
For an $S \subseteq [d]$, let $\vecy^{S} = \{m_1,\ldots,m_r\}$ and $\vecy^{-S} = \{n_1,\ldots,n_t\}$ in some order. Define $M_{f,S} \in \F^{r \times t}$ such that $M_{f,S}(i,j)$ is the coefficient of $n_j$ in $\Coeff{m_i}(f)$. Note that the $i^{th}$ row of $M_{f,S}$ is the polynomial $\Coeff{m_i}(f)$ written in the coefficient vector form.

For a type $\tau$ in a tree $T$, $S_{\tau}$ will denote the set of leaves of the node $\tau$ in $T$. Consequently, we will also use just $M_{f,\tau}$ to mean $M_{f,S_{\tau}}$. We will denote by $B_{f,\tau}$ a set of monomials from $\vecy^{S_{\tau}}$ such that the rows indexed by them in $M_{f,S}$ will form a basis of the rows of $M_{f,S}$. Note that if $\tau$ has children $\tau_1,\tau_2$, then we can ensure that our choice of $B_{f,\tau}$ satisfies $B_{f,\tau} \subseteq B_{f,\tau_1} \times B_{f,\tau_2}$ as the latter is clearly a spanning set. Using such a basis $B_{f,\tau}$, we can then write down a set of dependencies as below corresponding to $f$ and $\tau$.
\begin{equation}
\label{eq:dependency}
\forall m \in \vecy^{S_{\tau}} : \Coeff{m}(f) = \sum_{m' \in B_{f,\tau}} \gamma_{m,m'} \Coeff{m'}(f).
\end{equation}

\noindent
Using this, we can rewrite $f$ in the following way for any $\tau \in T$.
\begin{align}
  f &= \sum_{m_k \in \vecy^{S_{\tau}}} m_k \inparen{\sum_{m'_i \in B_{f,\tau}} \gamma_{i,k} \Coeff{m'_i}(f)} = \sum_{m'_i \in B_{f,\tau}} \inparen{\sum_{m_k \in \vecy^{S_{\tau}}} \gamma_{i,k} m_k} \Coeff{m'_i}(f) \nonumber \\
f &= \sum_{m'_i \in B_{f,\tau}} u_i(\vecy^{S_{\tau}}) \Coeff{m'_i}(f) \quad\text{for some $u_i \in \F[\vecy_{S_\tau}]$}\label{eq:sumBasisCoeffs}.
\end{align}

\noindent
Suppose $\tau \in T$ has two children $\tau_1$ and $\tau_2$ that share the same dependencies for $g$ as well. That is,

\begin{minipage}{0.4\textwidth}
\begin{align*}
  f &= \sum_{m'_i \in B_{f,\tau_1}} u_i(\vecy^{S_{\tau_1}}) \Coeff{m'_i}(f),\\
  g &= \sum_{m'_i \in B_{f,\tau_1}} u_i(\vecy^{S_{\tau_1}}) \Coeff{m'_i}(g),\\
\end{align*}
\end{minipage}
\hfill
\begin{minipage}{0.4\textwidth}
\begin{align*}
  f &= \sum_{n'_j \in B_{f,\tau_2}} v_j(\vecy^{S_{\tau_2}}) \Coeff{n'_j}(f),\\
  g &= \sum_{n'_j \in B_{f,\tau_2}} v_j(\vecy^{S_{\tau_2}}) \Coeff{n'_j}(g).\\
\end{align*}
\end{minipage}

\noindent
Combining them (and renaming the variables by dropping the ${}'$s), we get 
\begin{align*}
  f & = \sum_{(m_i,n_j) \in B_{f,\tau_1}\times B_{f,\tau_2}} u_i(\vecy^{S_{\tau_1}}) v_j(\vecy^{S_{\tau_2}}) \cdot \Coeff{m_i \cdot n_j}(f),\\
  g & = \sum_{(m_i,n_j) \in B_{f,\tau_1}\times B_{f,\tau_2}} u_i(\vecy^{S_{\tau_1}}) v_j(\vecy^{S_{\tau_2}}) \cdot \Coeff{m_i \cdot n_j}(g).
\end{align*}

\noindent
Observe that if for all $m \in B_{f,\tau_1} \times B_{f,\tau_2}$ we have
\[
    \Coeff{m}(f) = \sum_{m' \in B_{f,\tau}} \gamma_{m,m'} \Coeff{m'}(f)\qquad,\qquad
    \Coeff{m}(g)  = \sum_{m' \in B_{f,\tau}} \gamma_{m,m'} \Coeff{m'}(g),
\]
then this also forces that by \eqref{eq:sumBasisCoeffs}, for $\tau$:
\[
  f = \sum_{m_i \in B_{f,\tau}} u_i'(\vecy^{S_{\tau}}) \Coeff{m_i}(f) \qquad,\qquad
  g = \sum_{m_i \in B_{f,\tau}} u_i'(\vecy^{S_{\tau}}) \Coeff{m_i}(g).
\]
\noindent Since $g$ is \emph{not} computable by a $\ComUPT$ circuit with underlying shape $T$ this cannot happen for all $\tau \in T$. Let us pick the lowest $\tau$ (closest to the leaves; and say its children are $\tau_1,\tau_2$) such that for some $m \in B_{f,\tau_1}\times B_{f,\tau_2}$ we have
\begin{align}
  \label{eqn:f-g-dependency-diff}
  \begin{split}
    \Coeff{m}(f) &= \sum_{m' \in B_{f,\tau}} \gamma_{m,m'} \Coeff{m'}(f),\\
    \Coeff{m}(g) & \neq \sum_{m' \in B_{f,\tau}} \gamma_{m,m'} \Coeff{m'}(g).
  \end{split}
\end{align}

\noindent
The choice of the vector of polynomials is now clear. If $w' = \abs{B_{f,\tau_1}}\cdot \abs{B_{f,\tau_2}} \leq w^2$, then
\begin{align*}
  R & := \inparen{u_i(\vecy^{S_{\tau_1}}) v_j(\vecy^{S_{\tau_2}}) \;:\; (m_i,n_j) \in B_{f,\tau_1}\times B_{f,\tau_2}} \in \F[\vecy_{S_\tau}]^{1\times w'}\\
  P & := \inparen{\Coeff{m_i \cdot n_j}(f) \;:\; (m_i,n_j) \in B_{f,\tau_1}\times B_{f,\tau_2}}^T \in \F[\vecy_{-S_\tau}]^{w'\times 1}\\
  Q & := \inparen{\Coeff{m_i \cdot n_j}(g) \;:\; (m_i,n_j) \in B_{f,\tau_1}\times B_{f,\tau_2}}^T \in \F[\vecy_{-S_\tau}]^{w'\times 1}.
\end{align*}
\noindent
It is clear from the definition that the vectors $P$ and $Q$ are made up of coefficients of $f$ and $g$. Also, \eqref{eqn:f-g-dependency-diff} provides a suitable vector $\Gamma$ of support at most $w + 1$ such that $\Gamma P = 0$ but $\Gamma Q \neq 0$. 

It follows that the coefficient space of $R$ is full-rank as the sets of polynomials $\setdef{u_i}{i\in B_{f,\tau_1}}$ and $\setdef{v_j}{j\in B_{f,\tau_2}}$ are linearly independent and are on disjoint sets of variables.

We only need to show that every entry of $R$ can also be computed by a $\ComUPT$ circuit of preimage-width at most $w^2$. To see this, observe that the set of polynomials $\setdef{u_i(\vecy^{S_{\tau_1}})}{i\in B_{f,\tau_1}}$ spans the set $\setdef{\coeff_m(f)}{m \in \vecy^{-S_{\tau_1}}}$, and similarly $\setdef{v_j(\vecy^{S_{\tau_2}})}{j\in B_{f,\tau_2}}$ spans $\setdef{\coeff_n(f)}{n \in \vecy^{-S_{\tau_2}}}$. Since the dimension of these spaces is at most $w$, it follows that each $u_i(\vecy^{S_{\tau_1}})$ can be written as a linear combination of at most $w$ many $\coeff_m(f)$'s, and similarly each $v_j(\vecy^{S_{\tau_2}})$. \autoref{obs:UPTSMLforCoeff} shows that each of the coefficient polynomials can also be computed by $\ComUPT$ circuits of preimage-width at most $w$. Thus, by computing each of the $u_i$'s and $v_j$'s separately, and then taking all $w^2$ products, we have a $\ComUPT$ circuit of preimage-width at most $w^2$ that simultaneously computes all the entries of $R$. 
\end{proof}

\end{document}

%% file: thmmacros.tex
\usepackage{amsthm}
\usepackage{thmtools,thm-restate}

\numberwithin{equation}{section}
\declaretheoremstyle[bodyfont=\it,qed=\qedsymbol]{noproofstyle}

\declaretheorem[numberlike=equation]{observation}
\declaretheorem[numberlike=equation,style=noproofstyle,name=Observation]{observationwp}
\declaretheorem[name=Observation,numbered=no]{observation*}

\declaretheorem[numberlike=equation]{theorem}

\declaretheorem[name=Theorem,numbered=no]{theorem*}

\declaretheorem[numberlike=equation]{lemma}
\declaretheorem[name=Lemma,numbered=no]{lemma*}
\declaretheorem[numberlike=equation,style=noproofstyle,name=Lemma]{lemmawp}

\declaretheorem[numberlike=equation]{corollary}
\declaretheorem[name=Corollary,numbered=no]{corollary*}

\declaretheorem[name=Proposition,numbered=no]{proposition*}

\declaretheorem[numberlike=equation]{claim}
\declaretheorem[name=Claim,numbered=no]{claim*}

\declaretheorem[name=Conjecture,numbered=no]{conjecture*}

\declaretheorem[name=Question,numbered=no]{question*}

\declaretheoremstyle[bodyfont=\it,qed=$\lozenge$]{defstyle} 

\declaretheorem[numberlike=equation,style=defstyle]{definition}
\declaretheorem[unnumbered,name=Definition,style=defstyle]{definition*}

\declaretheorem[unnumbered,name=Example,style=defstyle]{example*}

\declaretheorem[unnumbered,name=Notation=defstyle]{notation*}

\declaretheorem[unnumbered,name=Construction,style=defstyle]{construction*}

\declaretheorem[unnumbered,name=Remark,style=defstyle]{remark*}


%% file: lazymacros.tex
\newcommand{\coeff}{\operatorname{coeff}}

\newcommand{\Coeff}[1]{\operatorname{coeff}_{#1}}
\newcommand{\ComUPT}{\operatorname{UPT\text{-}SML}}

\newcommand{\SumUPT}[1]{\Sigma^{#1}\operatorname{-UPT}}
\newcommand{\ComSumUPT}[1]{\Sigma^{#1}\operatorname{-UPT-SML}}

\newcommand{\Fspan}{\operatorname{span}}

\newcommand{\indentBlock}[1]{\begin{addmargin}[2.5em]{2.5em} #1 \end{addmargin}}

\newcommand{\eip}{\operatorname{eip}}

\newcommand{\wt}{\operatorname{wt}}
\newcommand{\Mons}{\operatorname{Mons}}

\newcommand{\Pal}{\operatorname{Pal}}
\newcommand{\MPal}{\operatorname{Pal}^{\mathrm{mov}}}

%% file: bibmacros.tex
\usepackage{nth}
\usepackage{intcalc}
\usepackage{etoolbox}
\usepackage{xstring}
\hypersetup{
breaklinks=true 
}

\newcommand{\ECCCurl}[2]{http://eccc.hpi-web.de/report/\ifnumcomp{#1}{>}{93}{19}{20}#1/#2/}

\newcommand{\shortECCC}[2]{\texttt{\href{http://eccc.hpi-web.de/report/\ifnumcomp{#1}{>}{93}{19}{20}#1/#2/}{eccc:TR#1-#2}}}

\newcommand{\parseECCC}[1]{
\StrSubstitute{#1}{TR}{}[\tmpstring]%
\IfSubStr{\tmpstring}{/}{ 
\StrBefore{\tmpstring}{/}[\ecccyear]%
\StrBehind{\tmpstring}{/}[\ecccreport]%
}{
\StrBefore{\tmpstring}{-}[\ecccyear]%
\StrBehind{\tmpstring}{-}[\ecccreport]%
}%
\shortECCC{\ecccyear}{\ecccreport}}